\documentclass[12pt,oneside]{article}

\pdfoutput=1

\usepackage[margin=1in]{geometry}

\usepackage[
  bookmarks=true,
  bookmarksnumbered=true,
  bookmarksopen=true,
  pdfborder={0 0 0},
  breaklinks=true,
  colorlinks=true,
  linkcolor=black,
  citecolor=black,
  filecolor=black,
  urlcolor=black,
]{hyperref}

\usepackage{isomath}
\usepackage{fnbreak}
\usepackage{amsthm,amsmath,amssymb}
\usepackage[round]{natbib}

\usepackage{tikz}
\usetikzlibrary{shapes.geometric, arrows}
\tikzstyle{internal} = [thin,rectangle, rounded corners, minimum width=1cm, minimum height=1cm,text centered, draw=black]
\tikzstyle{arrow} = [thick,->,>=stealth]

\theoremstyle{plain}
\newtheorem{theorem}{Theorem}[section]
\newtheorem{lemma}[theorem]{Lemma}
\newtheorem{remark}[theorem]{Remark}

\theoremstyle{definition}
\newtheorem{definition}[theorem]{Definition}

\newcommand{\eqdef}{:=}
\newcommand{\scf}{\mathit{scf}}
\newcommand{\Path}{\mathit{Path}}

\hypersetup{
  pdftitle       = {Gibbard-Satterthwaite Success Stories and Obvious Strategyproofness},
  pdfauthor      = {Sophie Bade <sophie.bade@rhul.ac.uk> and Yannai A. Gonczarowski <yannai@gonch.name>},
}

\title{Gibbard-Satterthwaite Success Stories\texorpdfstring{\\}{ }and Obvious Strategyproofness}
\author{Sophie Bade \and Yannai A.\ Gonczarowski \thanks{First online draft: October 2016.
Bade: Department of Economics, Royal Holloway University of London; and Max Planck Institut for Research on Collective Goods, Bonn, \emph{e-mail}: \href{mailto:sophie.bade@rhul.ac.uk}{sophie.bade@rhul.ac.uk}.
Gonczarowski: Einstein Institute of Mathematics, Rachel \& Selim Benin School of Computer Science \& Engineering, and Federmann Center for the Study of Rationality, The Hebrew University of Jerusalem; and Microsoft Research, \emph{e-mail}: \href{mailto:yannai@gonch.name}{yannai@gonch.name}.
We thank Sergiu Hart, Shengwu Li, Jordi Mass\'{o}, Ahuva Mu'alem, Noam Nisan, Marek Pycia, and Peter Troyan for their comments.
This collaboration is supported by the ARCHES Prize.
Yannai Gonczarowski is supported by the Adams Fellowship Program of the Israel Academy of Sciences and Humanities; his work is supported by ISF grants 230/10 and 1435/14 administered by the Israeli Academy of Sciences, and by Israel-USA Bi-national Science Foundation (BSF) grant 2014389.}}
\date{March 18, 2017}

\begin{document}

\maketitle

\begin{abstract}
The Gibbard-Satterthwaite Impossibility Theorem \citep{Gibbard1973,Satterthwaite1975} holds that dictatorship is the only Pareto optimal and strategyproof social choice function on the full domain of preferences. Much of the work in mechanism design aims at getting around this impossibility theorem. Three grand success stories stand out.
On the domains of single peaked preferences, of house matching, and of quasilinear preferences, there are appealing Pareto optimal and strategyproof social choice functions.
We investigate whether these success stories are robust to strengthening strategyproofness to obvious strategyproofness, recently introduced by \citet{Li2015}. A social choice function is obviously strategyproof (OSP) implementable if even cognitively limited agents can recognize their strategies as weakly dominant.

For single peaked preferences, we characterize the class of OSP-implementable and unanimous social choice functions as \emph{dictatorships with
safeguards against extremism} --- mechanisms (which turn out to also be Pareto optimal) in which the dictator can choose the outcome, but other
agents may prevent the dictator from choosing an outcome that is too extreme. Median voting is consequently not OSP-implementable. Indeed, the only OSP-implementable quantile rules choose either the minimal or the maximal ideal point.
For house matching, we characterize the class of
OSP-implementable and Pareto optimal matching rules as \emph{sequential barter with lurkers} --- a significant
generalization over bossy variants of bipolar serially dictatorial rules. While \citet{Li2015} shows that second-price auctions are OSP-implementable when only one good is sold, we show that this positive result does not extend to the case of multiple goods. Even when all agents' preferences over goods are quasilinear and additive, no welfare-maximizing auction where losers pay nothing is OSP-implementable when more than one good is sold. Our analysis makes use of a gradual revelation principle, an analog of the (direct)
revelation principle for OSP mechanisms that we present and prove.
\end{abstract}

\section{Introduction}

The concern with incentives sets mechanism design apart from algorithm and protocol design. A mechanism that directly elicits preferences is strategyproof if no agent ever has any incentive to misreport her preferences. Strategyproofness may, however, not be enough to get the participants in a mechanism to report their true preferences.
Indeed, the participants must \emph{understand} that it is in their best interest to reveal their true preferences --- they must understand that the mechanism is strategyproof.
Depending on whether it is more or less easy to grasp the strategic properties of a mechanism, different extensive forms that implement the same strategyproof mechanism may yield different results in practice: while the participants may in some case easily understand that no lie about their preferences can possibly benefit them, they may not be able to see this in a different extensive form that implements the same social choice function.

Think of a second-price auction, for example. We can on the one hand solicit sealed bids, and award the auctioned good to the highest bidder, charging her the second-highest bid. Alternatively, we may use a clock that continuously increases the price of the good. In this case, agents choose when to drop out, and once only one last agent remains, this agent obtains the good and pays the current clock price. Assuming that the bidders' values are independent, both mechanisms implement the same --- strategyproof --- social choice function: submitting one's true value in the sealed-bid auction, and equivalently, dropping out of the ascending clock auction when one's own value is reached, are weakly dominant strategies in these two mechanisms.
However, it is well documented that agents approach these two mechanisms differently. It appears~\citep{KHL1987} that the strategyproofness of the implementation using an ascending clock is easier to understand than the strategyproofness of the sealed-bid implementation. Recently, \citet{Li2015} proposed the concept of \textbf{obvious strategyproofness}, which captures this behavioral difference.

Unlike classic strategyproofness, which is a property of the social choice function in question, \citeauthor{Li2015}'s obvious strategyproofness is a property of the mechanism implementing this social choice function.
To check whether a strategy is \textbf{obviously dominant} for a given player, one must consider each of the histories at which this player can get to choose in case she follows the given strategy. Fixing any such history, the player compares the worst-possible outcome starting from this history given that she follows this strategy, with the best-possible outcome from this history given that she deviates at the history under consideration. To evaluate these best- and worst-possible outcomes, the player considers all possible choices of all other players in the histories following on the current history.  If at each such history, the worst-possible outcome associated with following the strategy is no worse than the best-possible outcome associated with a deviation, then the strategy is said to be obviously dominant. If each player has an obviously dominant strategy, then the mechanism is obviously strategyproof (OSP), and the social choice function that it implements is \textbf{OSP-implementable}. OSP-implementability is a stricter condition than strategyproofness. \citet{Li2015} shows that even cognitively limited agents may recognize an obviously strategyproof mechanism as strategyproof. 

\citet{Li2015} shows that the implementation of second-price auction via an ascending clock is obviously strategyproof, while the implementation via sealed bids is not. To see this, consider a bidder with value $4$. If she submits a sealed bid of $4$, then the worst-case utility she may obtain is~$0$ (if her bid is not the highest). If she instead were to bid $6$ --- and all other bidders bid $0$ --- then she would obtain a utility of $4$. So, bidding her true value is not obviously dominant. In contrast, when the clock in an ascending implementation stands at $3$, this same agent compares the worst utility associated with dropping out with the best utility associated with staying in. Since both equal $0$, staying in is an obviously dominant choice at this history (as well as at any other history where the clock stands at less than $4$).

While \citet{Li2015} makes a strong case that obviously strategyproof mechanisms outperform mechanisms that are only strategyproof, he leaves open the question of which social choice functions are OSP-implementable.
The current paper examines this question through the lens of the popular desideratum of Pareto optimality. That is, this paper asks which Pareto optimal social choice functions are OSP-implementable.

When agents may hold any preference over a set of at least three outcomes, then any
strategyproof and Pareto optimal social choice function is --- by the Gibbard-Satterthwaite Impossibility Theorem \citep{Gibbard1973,Satterthwaite1975} --- dictatorial. So, to find Pareto optimal, OSP-implementable and nondictatorial social choice functions, we must investigate social choice functions for domains that are not covered by the Gibbard-Satterthwaite theorem. We accordingly conduct our analysis in the three most popular domains that provide ``escape routes'' from the Gibbard and Satterthwaite impossibility theorem: the domain of single peaked preferences, the quasilinear domain, and the house matching domain. On each of these three domains, there are some well-studied strategyproof, Pareto optimal, and nondictatorial social choice functions.

In each of these three domains, we fully characterize the class of OSP-implementable and Pareto optimal social choice function (for the quasilinear domain, as is customary, we also require that losers pay nothing).
On the one hand, our findings suggest that obvious strategyproofness is a highly restrictive concept.
Indeed, apart from two special cases of ``popular'' mechanisms that were already known to be OSP-implementable --- the auction of a single good~\citep{Li2015}, and trade with no more than two traders at any given round~\citep{AG2015,Troyan2016} --- our analysis of all three domains finds only one more such special case of ``popular'' mechanisms: choosing the maximum or minimum ideal point when all agents have single peaked preferences. On the other hand, our complete characterizations show that outside of these, a few rather exotic and quite intricate mechanisms are also obviously strategyproof.

The investigation of each of these three domains builds on a revelation principle that we state and prove for obviously strategyproof mechanisms. This revelation principle shows that a social choice function is OSP-implementable if and only if it can be implemented by an \textbf{obviously incentive compatible gradual revelation mechanism}. A mechanism is a gradual revelation mechanism if each choice of an agent is identified with a partition of the set of all the agent's preferences that are consistent with the choices made by the agent so far. In a truthtelling strategy, the agent gradually reveals her preference: at each juncture, she chooses a smaller set of preferences that her own preference belongs to. With her last choice, the agent fully reveals her preference. Furthermore, in a gradual revelation mechanism, whenever an agent can fully disclose her preference without hurting obvious strategyproofness, she does so. A gradual revelation mechanism does, moreover, not allow for simultaneous moves, for directly consecutive choices by the same agent, or for choice sets with a single action.

In the domain of single peaked preferences, we find that a mechanism is Pareto optimal and obviously strategyproof if and only if it is a \textbf{dictatorship with safeguards against extremism}. In such a mechanism, there is one ``dictator'' who may freely choose any option from a central set. If she would rather choose an option to the right or left of that central set, then she needs to win the approval of some other agents. The set of agents whose approval is needed for right-wing positions increases as these positions move farther to the right. The same holds for left-wing positions. Finally, if the electorate has already identified that one of two adjacent options will be chosen, then a process of arbitration between these two options may ensue.
Dictatorships with safeguards against extremism embed dictatorships: in the case of a dictatorship, the central set from which the dictator may freely choose is the grand set of all options. Dictatorships with safeguards against extremism also embed social choice functions that choose the minimal (and respectively maximal) ideal point of all agents. However, median voting is not OSP-implementable. To see this, suppose it were, and consider the first agent to make any decision in an obviously strategyproof mechanism that implements median voting. For any deviation from the truthful revelation of her ideal policy, the best case for the agent is that all other agents were to announce her own ideal policy as theirs, and this policy would get chosen.
Conversely, if the agent follows the truthtelling strategy, then the worst case for the agent is that all other agents declare the policy that this agent considers worst as their ideal policy.

For the quasilinear domain with multiple goods, we find that any Pareto optimal (or equivalently, welfare maximizing) mechanism in which losers pay nothing (such as VCG with the Clarke pivot rule~ \citep{Vickrey1961,Clarke1971,Groves1973}) is not obviously strategyproof. To make our case strongest, we show that this holds even if there are only two goods and all agents' utilities are additive. This implies that \citeauthor{Li2015}'s \citeyear{Li2015} result that a second-price auction is OSP-implementable does not extend beyond one good. To get some intuition into the restrictiveness of obvious strategyproofness in the setting of auctions, consider two sequential ascending auctions: the first for a bottle of wine and the second for a violin. Assume that apart from a single agent who participates in both auctions, all other participants participate in only one of the auctions, and furthermore, those who participate only in the second auction (for the violin) have knowledge of neither the bids nor the outcome of the wine auction. Assume that the utility of the single agent who participates in both auctions is additive, so this agent values the bundle consisting of both the bottle and the violin at the sum of her values for the bottle alone and for the violin alone. We emphasize that in the wine auction, this agent considers the other agents' behavior in all later histories, including the histories of the violin auction. Observe that if this agent values the bottle at~$4$ then, in contrast with the setting of a single ascending auction, she may not find it obviously dominant to continue bidding at $3$. Indeed, if she drops out at $3$ and if all agents behave in the most favorable way in all later histories, then she gets the violin for free; otherwise, she may be outbid for the violin. So, if she values the violin at $v>1$, then staying in the wine auction at $3$ is not obviously dominant.

For the house matching domain, we find that a mechanism is Pareto optimal and obviously strategyproof if and only if it can be represented as \textbf{sequential barter with lurkers}. Sequential barter is a trading mechanism with many rounds. At each such round, there are at most two owners. Each not-yet-matched house sequentially becomes owned by one of them. Each of the owners may decide to leave with a house that she owns, or they may both agree to swap. If an owner does not get matched in the current round, she owns at least the same houses in the next round. When a \textbf{lurker} appears, she may ultimately get matched to any one house in some set $S$. A lurker is similar to a dictator in the sense that she may immediately appropriate \emph{all but one special house} in the set $S$. If she favors this special house the most, she may ``lurk'' it, in which case she is no longer considered an owner (so there are at most two owners, and additionally any number of lurkers, each for a different house). If no agent who is  entitled to get matched with this special house chooses to do so, then the lurker obtains it as a residual claimant. Otherwise, the lurker gets the second-best house in this set $S$.
The definition of sequential barter with lurkers reveals that the various mechanics that come into play within obviously strategyproof mechanisms are considerably richer and more diverse than previously demonstrated.

The paper is organized as follows. Section~\ref{sec: definitions} provides the model and definitions, including the definition of obvious strategyproofness.
Section~\ref{sec: revelation} presents the gradual revelation principle. Section~\ref{sec: voting} studies voting with two possible outcomes. Sections~\ref{sec: single peaked}, \ref{sec: quasilinear}, and~\ref{sec: house-matching} respectively study single peaked preferences, quasilinear preferences, and house matching. We conclude in Section~\ref{sec: conclusion}. Proofs and some auxiliary results are relegated to the appendix.

\section{Model and Definitions}\label{sec: definitions}

\subsection{The Design Problem}

There is a finite set of agents $N\eqdef\{1,\ldots, n\}$ with typical element $i\in N$ and a set of outcomes~$Y$. Agent $i$'s preference $R_i$
is
drawn from a set of possible preference $\mathcal{R}_i$. Each possible preference $R_i$ is a complete and transitive order on $Y$, where $x P_i
y$
denotes the case that $xR_iy$ but not $y R_i x$ holds. If $xR_iy$ and $yR_ix$, then $x$ and $y$ are $R_i$-indifferent.
The domain of all agents' preferences is $\mathcal{R}\eqdef\mathcal{R}_1\times\cdots\times\mathcal{R}_n$. If two alternatives $x$ and $y$ are
$R_i$-indifferent for each $R_i\in \mathcal{R}_i$, then
$x$ and $y$ are \textbf{completely $\mathbfit{i}$-indifferent}.
 The set of all outcomes that are completely $i$-indifferent to $y$ is $[y]_i$.

We consider three classes of design problems.
In a political problem with \textbf{single peaked preferences}, we represent the
set of social choices as a the set of integers~$\mathbb{Z}$. For every $i\in N$ and any $R_i\in \mathcal{R}_i$, there exists an
\textbf{ideal point} $y^*\in Y$ such that $y'< y\leq y^*$ or $y^*\geq y > y'$ implies   $yP_i y'$ for all $y,y'\in Y$.
In the case of \textbf{quasilinear preferences}, the outcome space is $Y\eqdef X\times M$, where $X$ is a set of allocations and $M$ is a set of
monetary payments with $m_i$ representing the payment charged from agent $i$. Each agent's preference is represented by a utility function
$U_i(x,m)= u_i(x)+m_i$, where $u_i$ is a utility function on $X$.
In a \textbf{house matching
problem}, the outcome space $Y$ is the set of one-to-one matchings between agents and a set $O$ of at least as many\footnote{We show in the appendix that our results extend to the case where some agents may not be matched to any house, i.e., the case of matching with outside options. In this case, there is no restriction on the number of houses.} houses, constructed as follows. An agent-house pair $(i,o)$ is a match, and a
matching is a set of matches where each agent $i\in N$ partakes in precisely one match, and no house $o\in O$ partakes in multiple matches.
Each agent only cares about the house she is matched with. Each agent, moreover, strictly ranks any two different houses.
So any $x$ and $y$ are completely $i$-indifferent if and only if $i$ is matched to the same house under $x$ and under $y$.

A social choice function $\scf:\mathcal{R}\to Y$ maps each profile of preferences $R\in\mathcal{R}$ to an outcome $\scf(R)\in Y$.

\subsection{Mechanisms}

A (deterministic) \textbf{mechanism} is an extensive game form with the set $N$ as the set of players. The set of \textbf{histories} $H$ of this
game form are the set of (finite and infinite\footnote{Unlike \citet{Li2015}, we allow for infinite histories mainly to allow for easier
exposition of our
analysis of the domain of single peaked preferences. See Section~\ref{sec: single peaked} for more details.}) paths from the root of the directed
game form tree. For
a history $h=(a^k)_{k=1,\ldots}$ of length at least $L$, we denote by $h|_L=(a^k)_{k=1,\ldots,L}\in H$ the \textbf{subhistory} of $h$ of length
$L\ge0$. We write $h'\subseteq h$ when $h'$ is a subhistory of $h$.
A history is \textbf{terminal} if it is not a subhistory of any other history. (So a terminal history is either a path to a leaf or
an infinite path.)
The set of all terminal histories is $Z$.

The set of possible actions
after the nonterminal history $h$ is  $A(h)\eqdef\{a\mid (h,a)\in H\}$.
The player function $P$ maps any nonterminal history $h\in H\setminus Z$ to a player $P(h)$ who gets to choose from all actions $A(h)$ at $h$.
Each terminal history $h\in Z$ is mapped to an outcome in~$Y$.

Each player $i$ has an \textbf{information partition} $\mathcal{I}_i$ of the set $P^{-1}(i)$ of all nodes $h$ with $P(h)=i$, with $A(h)=A(h')$
if $h$ and $h'$ belong to the same cell  of $\mathcal{I}_i$. The cell to which $h$  with $P(h)=i$ belongs is $I_i(h)$.\footnote{\citet{Li2015}
also imposes the condition of perfect recall onto information partitions.
Our results hold
with and without prefect recall. For ease of exposition, we therefore do not impose prefect recall.}
A \textbf{behavior}  $B_i$
for player $i$
is an $\mathcal{I}_i$-measurable function mapping each  $h$ with $P(h)=i$ to an action in $A(h)$. A \textbf{behavior profile} $B=(B_i)_{i\in N}$
lists a behavior for each player. The set of behaviors for player $i$ and the set of behavior profiles are respectively denoted $\mathcal{B}_i$
and $\mathcal{B}$.
A behavior profile~$B$ induces a unique terminal history $h^B=(a^k)_{k=1,\ldots}$ s.t.\ $a^{k+1}=B_{P(h^B|_k)}(h^B|_k)$ for every $k$ s.t.\ $h^B|_k$
is
nonterminal.
 The \textbf{mechanism} $M:\mathcal{B}\to Y$  maps the behavior profile $B\in \mathcal{B}$ to the outcome $y\in Y$ that is associated with the
 terminal history $h^B$.
We call the set of all subhistories of the terminal history $h^B$ the \textbf{path} $\Path(B)$.
A \textbf{strategy} $\mathbf{S}_i$ for agent $i$ is a function $\mathbf{S}_i:\mathcal{R}_i\to \mathcal{B}_i$. The \textbf{strategy profile} $\mathbf{S}=(\mathbf{S}_i)_{i\in N}$ induces
the social choice function $\scf:\mathcal{R}\to Y$ if $\scf(R)$ equals $M(\mathbf{S}(R))$ for each $R\in \mathcal{R}$.

In a \textbf{direct revelation mechanism} all agents move simultaneously. Agent $i$ 's behavior space consists of his set of possible preferences
$\mathcal{R}_i$.
A strategy $\mathbf{S}_i$ for $i$ maps each preference $R_i\in \mathcal{R}_i$ to another preference $\mathbf{S}_i(R_i)$. The \textbf{truthtelling strategy} $\mathbf{T}_i$ maps each preference $R_i$ onto itself.

\subsection{Normative Criteria}

A social choice function $\scf$ is \textbf{Pareto optimal} if it maps any $R$ to an outcome $\scf(R)$ that is Pareto optimal at
$R$. An
outcome $y\in Y$ in turn is Pareto optimal at $R$ if there exists no $y'\in Y$ such that $y'R_i y$ holds for all $i$ and
$y'P_{i'}y$ holds for at least one $i'$.

A strategy $\mathcal{S}_i$ in a mechanism $M$ is \textbf{dominant} if $M(\mathcal{S}_i(R_i),B_{-i})R_iM(B)$ holds for all behavior profiles $B$ and all $R_i\in \mathcal{R}_i$. So $\mathcal{S}_i$ is dominant if it prescribes  for each possible preference $R_i\in \mathcal{R}_i$ a behavior $\mathbf{S}_i(R_i)$ such that $i$ prefers the outcome of $M$ given that behavior to the outcome of $M$ given any other behavior $B_i$, no matter which behavior the other agents follow.
A direct revelation mechanism is \textbf{incentive compatible} if truthtelling is a dominant strategy for each player.
The revelation principle states that each social choice function that can be implemented in dominant strategies can be implemented by an
incentive compatible direct revelation mechanism.

A strategy $\mathbf{S}_i$ is \textbf{obviously dominant} \citep{Li2015} for agent $i$  if for every $R_i\in \mathcal{R}_i$, behavior profiles $B$ and $B'$, and histories $h$ and $h'$ with $h\in\Path(\mathcal{S}_i(R_i),B_{-i})$, $h'\in\Path(B')$, $P(h)=P(h')=i$, $\mathcal{I}_i(h)=\mathcal{I}_i(h')$, and $\mathcal{S}_i(R_i)(h)\neq B'_i(h')$ we have
\begin{eqnarray*}
M(\mathcal{S}_i(R_i),B_{-i})R_iM(B').
\end{eqnarray*}
So, the strategy $\mathcal{S}_i$ has to meet a stricter condition to be considered not just strategyproof but also obviously strategyproof:
at each juncture that is possibly reached during the game, agent $i$ considers whether to deviate from the action $\mathcal{S}_i(R_i)(h)$ prescribed by his strategy~$\mathcal{S}_i$ at that juncture to a different action $B_i(h)$. The condition that $\mathcal{S}_i$ has to meet is that even under the worst-case scenario (minimizing over all other agents' behaviors and over $i$'s uncertainty) if agent $i$ follows $\mathcal{S}_i(R_i)(h)$ at that juncture, and under the best-case scenario (maximizing over all other agents' behaviors and over $i$'s uncertainty) if agent $i$ deviates to $B_i(h)\neq \mathcal{S}_i(R_i)(h)$, agent $i$ still prefers not to deviate.

A social choice function $\scf$ is implementable in obviously dominant strategies, or \textbf{OSP-imple\-mentable}, if $\mathbf{S}$ is a profile of obviously dominant strategies in some mechanism $M$
and if $\scf(\cdot)=M(\mathbf{S}(\cdot))$.
In the next section, we show that a modified revelation principle holds for  implementation
in obviously dominant strategies.

\section{A Revelation Principle for\texorpdfstring{\\}{ }Extensive-Form Mechanisms}\label{sec: revelation}

In this section we develop an analogue, for OSP mechanisms, of the celebrated (direct) revelation principle (for strategyproof mechanisms). Our gradual revelation approach is conceptually similar to that of direct revelation: we define gradual revelation mechanisms so that agents gradually reveal more and more about their preferences. We  
then prove that any OSP-implementable social choice function is  implementable by an OSP gradual revelation mechanism. We use this gradual revelation principle throughout this paper.

A \textbf{gradual revelation mechanism} is a mechanism with the following additional properties:\footnote{\label{fn: gradual-revelation-properties}While most of the following properties are novel, \citet{AG2015} already showed that any OSP-implementable social choice function is also implementable by an OSP mechanism with Properties~\ref{prop: singleton-information} and~\ref{prop: labeled-actions}.  For completeness, we spell out the proof that Properties~\ref{prop: singleton-information} and~\ref{prop: labeled-actions} may be assumed without loss of generality.
\citet{PyciaTroyan2016} independently stated a property weaker than our Property~\ref{prop: force-asap-and-reveal}, and showed that it
may be assumed without loss of generality (see the discussion in Section~\ref{sec: house-matching} that relates that paper to ours).
}
\begin{enumerate}
\item\label{prop: singleton-information}
Each cell $I_i(h)$ of each information partition $\mathcal{I}_i$ is a singleton.
\item
No agent has two directly consecutive choices: $P(h)\neq P(h,a)$ holds for
every nonterminal history $(h,a)$.
\item
Choices are real: no $A(h)$ is a singleton.
\item\label{prop: labeled-actions}
Each finite history $h$ is identified with a nonempty set $\mathcal{R}_{i}(h)$  for each $i\in N$.
For the empty history, $\mathcal{R}_i(\emptyset)=\mathcal{R}_i$.
For each nonterminal $h$ with $P(h)=i$, the set $\{\mathcal{R}_i(h,a)\mid a\in A(h)\}$ partitions $\mathcal{R}_i(h)$. If $P(h)\neq i$, then $\mathcal{R}_i(h)=\mathcal{R}_i(h,a)$ for all $a\in A(h)$.
\item\label{prop: force-asap-and-reveal}
For every agent $i$, behavior $B_i$ for agent $i$, and nonterminal history $h$ with $i=P(h)$, if the set ${\{M(B) \mid B_{-i}~\mbox{s.t.}~h\in\Path(B)\}}$ is a nonempty set of completely $i$-indifferent outcomes, then $\mathcal{R}(h,B_i(h))$ is a singleton.
 \end{enumerate}

Property~\ref{prop: labeled-actions} requires that each agent with each choice reveals more about the set that her preference belongs to. Property~\ref{prop: force-asap-and-reveal} then  requires that  whenever the behavior of agent $i$ starting at $h$ ensures that the outcome lies in some given set of completely $i$-indifferent outcomes,
then  $i$ immediately ensures this with the action chosen at $h$.  Furthermore, $i$ completely reveals her preference when choosing this action.

A strategy $\mathbf{T}_i$ for player $i$ in a gradual revelation mechanism is a \textbf{truthtelling strategy} if $R_i\in
\mathcal{R}_i(h,\mathbf{T}_i(R_i)(h))$
holds for all nonterminal $h$ with $P(h)=i$ and all $R_i\in \mathcal{R}_i(h)$.  So $\mathbf{T}_i$ is a truthtelling strategy if
agent $i$ reveals which set his preference $R_i$  belongs to, whenever possible. If  $R_i\notin \mathcal{R}_i(h)$, then the definition imposes
no restriction on the behavior of agent $i=P(h)$  with preference $R_i$.
Since the specification of $\mathbf{T}_i(R_i)$ for histories $h$ with $R_i\notin \mathcal{R}_i(h)$ is inconsequential to our
analysis, we call any
truthtelling strategy  \emph{the} truthtelling strategy. A gradual revelation mechanism is
\textbf{obviously incentive compatible} if the truthtelling strategy $\mathbf{T}_i$ is obviously dominant for each agent $i$. We say that an
obviously incentive compatible gradual revelation mechanism $M$ implements a social choice function $\scf:\mathcal{R}\to Y$ if
$\scf(\cdot)=M(\mathbf{T}(\cdot))$.

\begin{theorem}\label{theorem: revelation}
A social choice function is OSP-implementable if and only if some obviously incentive compatible gradual
revelation mechanism $M$ implements it.
\end{theorem}

\noindent The proof of Theorem~\ref{theorem: revelation} is relegated to Appendix~\ref{app: revelation}.

For any $h$, we define the set $\mathcal{R}(h)$ as the set of all preference profiles $R\in\mathcal{R}$ with $R_i\in \mathcal{R}_i(h)$ for every $i$.
In a gradual revelation mechanism, $h$ is on the path $\Path(\mathbf{T}(R))$ if and only if $R\in \mathcal{R}(h)$. A gradual revelation mechanism is consequently obviously incentive compatible if and only if the following holds for each nonterminal history $h$ in $M$, where we denote $i=P(h)$:
\begin{eqnarray*}
M(\mathbf{T}(R))R_{i}M(\mathbf{T}(R'))\mbox{ for all }R,R'\in \mathcal{R}(h)\mbox{ s.t.\ $\mathbf{T}_{i}(R_{i})(h)\ne\mathbf{T}_{i}(R'_{i})(h)$}.
\end{eqnarray*}
So the agent $i$ who moves at $h$ must prefer the worst-case --- over all preference profiles of other agents such that $h$ is reached ---
outcome reached by truthtelling, i.e., by following $\mathbf{T}_i(R_i)$, over the best-case --- over all preference profiles of other agents such that $h$ is reached --- outcome reached by deviating to any alternative behavior that prescribes a different action $\mathbf{T}_i(R'_i)(h)\ne\mathbf{T}_i(R_i)(h)$ at $h$.

 For any history $h$, let $Y(h)$ be the set of all outcomes associated with a terminal history $h'$ with $h\subseteq h'$.
In an obviously incentive compatible gradual revelation mechanism $M$, let $h$ be a nonterminal history and let $i=P(h)$ . We define
$Y^*_h\subseteq Y(h)$ to be the set of outcomes $y$ such that there exists some $a\in A(h)$ s.t.\ $Y(h,a)\subseteq[y]_i$.
We define $A^*_h\subseteq A(h)$ to be the set of actions $a$
such that $Y(h,a)\subseteq
[y]_i$ for some $y\in Y^*_h$. We call $A^*_h$ the set of \textbf{dictatorial} actions at $h$. Let $\overline{A^*_h}\eqdef A(h)\setminus A^*_h$
and $\overline{Y^*_h}\eqdef Y(h)\setminus Y^*_h$. We call $\overline{A^*_h}$ the set of \textbf{nondictatorial} actions at $h$. We will show
below that
 $Y^*_h$ and $A^*_h$ are nonempty for the single peaked as well as the matching domain. (See Theorem~\ref{theorem: two options} and Lemmas~\ref{proof
 2}, \ref{proof 4}, ~\ref{proof 5}, and~\ref{lemma: one-undetermined}.) Before considering these domains, we perform some additional preliminary analysis in Appendix~\ref{app: revelation}.
 
\section{Voting}\label{sec: voting}

Majority voting is not obviously strategyproof even when there are just two possible outcomes, i.e., $Y\eqdef\{y,z\}$. In fact, unanimity (e.g., choosing the outcome $z$ if and only if all agents prefer it to $y$) is the only obviously strategyproof supermajority rule. In the sequential implementation of any other supermajority rule, the first agent $P(\emptyset)$ does not
have an action that determines one of the two choices. So, for whichever choice she picks, the worst-case scenario is that all other agents vote against her. On the other hand, the best-case scenario if she picks the other outcome is that all other agents would vote for her preferred outcome.

There are, however, some nondictatorial obviously strategyproof unanimous voting mechanisms. In a \textbf{proto-dictatorship}, each agent in a stream of agents is given either the choice between implementing $y$ or going on, or the choice between implementing $z$ and going on. The mechanism terminates either when one of these agents chooses to implement the outcome offered to her, or with a last agent who is given the choice between implementing $y$ or $z$.
 At each nonterminal history  $h$ of a proto-dictatorship $M$, (precisely) one of the following holds:
\begin{itemize}
\item
$Y^*_h=\{y\}$ and $\overline{A^*_h}=\{\tilde{a}\}$ with $Y(h,\tilde{a})=\{y,z\}$.
\item
$Y^*_h=\{z\}$ and $\overline{A^*_h}=\{\tilde{a}\}$ with $Y(h,\tilde{a})=\{y,z\}$, or
\item
$Y^*_h=\{y,z\}$ (and $\overline{A^*_h}=\emptyset$).
\end{itemize}
There is moreover no terminal history $h$ such that one agent moves twice on the path to reach $h$: 
 $P(h')\neq P(h'')$ holds for any $h'\subsetneq h''\subsetneq h$.

\begin{theorem}\label{theorem: two options}
Let $Y=\{y,z\}$. Then $M$ is obviously strategyproof and onto if and only if it is a proto-dictatorship.
\end{theorem}

The proof of Theorem~\ref{theorem: two options}, which readily follows from the analysis of Section~\ref{sec: revelation} (and Appendix~\ref{app: revelation}), is relegated to Appendix~\ref{app: voting}.

\section{Single Peaked Preferences}\label{sec: single peaked}

In the domain of single peaked preferences, the possible outcomes (also called policies) are $Y=\mathbb{Z}$, and each agent (also called voter)
has single peaked preferences, i.e., the agent prefers some $y \in \mathbb{Z}$, called the agent's \textbf{ideal point}, the most, and for every $y''>y'\ge y$ or $y''<y'\le y$, the agent strictly prefers
$y'$ over $y''$. A \textbf{unanimous} social choice function is one that, if the ideal points of all agents coincide, chooses the joint ideal point.
Unanimity is a strictly weaker assumption than Pareto optimality.

With single peaked preferences, there is a large range of strategyproof and unanimous social choice functions \cite{Moulin1980}. Most prominently,
median voting, which maps any profile of preferences to a median of all voters' ideal points is strategyproof and unanimous (and even Pareto
optimal). However, median voting is not obviously strategyproof when there are at least 3 voters. To see this, suppose some gradual revelation mechanism did implement median voting. 
Say the ideal point of the first agent in this mechanism is $y$ and 
 truthtelling prescribes for this agent to choose some action $a\in A(\emptyset)$. If all other voters declare their ideal point to be some $y'\neq y$, then $y'$, the median of all declared preferences, is implemented
 regardless of the first agent's choice $a$. If the first
agent deviates to some action $a'\neq a$ and if all voters --- according to the best-case scenario --- say their ideal point is $y$, then $y$ as 
the median of all announced preferences is implemented. In sum, truthtelling is not obviously strategyproof for the first agent.

A different, less popular, unanimous (and even Pareto optimal) and strategyproof social choice function for any single peaked domain is the
function $min$, which maps any profile of preferences to the minimal ideal point. We observe that if the set of possible ideal points is bounded from below by
some bound $\underline{y}$, then this function is OSP-implementable: The obviously strategyproof implementation of $min$ follows along the
lines of the (obviously strategyproof) ascending implementation of second-price auctions.
The $min$ mechanism starts with $\underline{y}$. For each number $y\in [\underline{y},\infty)\cap\mathbb{Z}$, sequentially in
increasing order, each agent is given an option to decide whether to continue or to stop. When one agent stops at some $y$, the mechanism
terminates with the social choice $y$.

In this section, we show that the obviously strategyproof and unanimous mechanisms for the domain of single peaked preferences are combinations of $min$, $max$, and
dictatorship, and furthermore, they are all Pareto optimal.
We define \textbf{dictatorship with safeguards against extremism}
for domains of single peaked preferences as follows: One agent, say  $1$, is called the dictator. All other agents have
limited veto rights. Specifically, each agent $i\neq 1$ can block extreme leftists policies and rightist policies in the rays
$(-\infty,l^i)$ and $(r^i,\infty)$, for some $l_i\le r_i \in \mathbb{Z}\cup\{-\infty,\infty\}$. Furthermore, there exists some $y^m$ with $l^i
\le y^m \le r^i$ for
all~$i$. Say that $\underline{y}^i$ and $\overline{y}^i$ respectively are agent $i$'s preferred policies in the rays
$(-\infty,l^i]$ and $[r^i,\infty)$. Then the outcome of the dictatorship with safeguard against extremism is
agent 1's most preferred policy in $\bigcap_{i\neq 1}[\overline{y}^i,\underline{y}^i]$.

According to this social choice function, agent 1 is free to choose any policy ``in the middle'':  If agent~1's ideal policy $y$ is in $[\max_{i\neq
1}{l^i},\min_{i\neq 1}r^i]$, then $y$ is implemented. Note that by assumption, $y^m\in [\max_{i\neq
1}{l^i},\min_{i\neq 1}r^i]$, and so this choice set is nonempty. If agent $1$'s ideal policy is farther to the left or right, then
it may only
be chosen if none of a select group of citizens vetoes this choice. As we consider more extreme policies, the group that needs to consent to the implementation of a policy increases. Dictatorships with safeguards against extremism embed standard dictatorships
($l^i=-\infty$ and $r^i=\infty$ for all $i$). They also embed $min$ when the ideal points are bounded from below by some $\underline{y}$ (by
$r^i=\underline{y}$ for all $i$) as well as $max$ when the ideal points are bounded from above by some $\overline{y}$ (by $l^i=\overline{y}$ for
all $i$).

Fix a dictatorship with safeguards against extremism $\scf$. Then $\scf$ is implementable in obviously dominant strategies by the mechanism that
first
 offers the dictator to choose any option in  $[L^*,H^*]\eqdef\bigcap_{i\neq 1}[\overline{y}^i,\underline{y}^i]$. If the
 dictator does not choose an option in this interval, then she indicates whether the mechanism is to continue to the left or to the right (according to
 the direction of the dictator's ideal point). If the
 mechanism continues to the right, then similarly to the implementation of $min$, the mechanism starts with $H^*$, and for each policy $y \in
 [H^*,\infty)$, sequentially in increasing order, each agent $i$ with either $i=1$ (the dictator) or $r^i \le y$ is given an option to decide
 whether to continue or to stop at $y$. When an agent stops at some $y$, the mechanism terminates with the social choice~$y$. If the
 mechanism continues to the left, then similarly to the implementation of $max$, the mechanism starts with $L^*$, and for each policy $y \in
 (-\infty,L^*]$ in decreasing order, the dictator (agent~1) and each agent $i$ with $l^i \ge y$ may decide to stop the mechanism with the implementation of $y$.

To see that this mechanism is obviously strategyproof, assume that the  ideal point of some agent~$i$ is $y^*$.
If this agent is the dictator and $y^*\in[L^*,H^*]$, then choosing $y^*$ as the outcome is obviously strategyproof, as the worst-case outcome under this strategy is the best-possible outcome. We claim that for any agent (whether or not the dictator), continuing to the right at any $y<y^*$ and then stopping at $y^*$ is obviously
strategyproof. Indeed, continuing at any $y<y^*$ is obviously strategyproof since the worst-case outcome under the strategy that continues until
$y^*$ and then stops at $y^*$ is in $[y,y^*]\cap\mathbb{Z}$, and therefore no worse than $y$, which is the best-possible outcome when deviating to
stopping at $y$. Stopping at $y=y^*$ is obviously strategyproof as it implements $i$'s top choice.
By the same argument, continuing to the left when $y>y^*$ and then stopping at $y^*$ is obviously strategyproof.

   To see that the
 mechanism is Pareto optimal, first consider the case where the dictator chooses a policy $y$ from $[L^*,H^*]$. In this case, the dictator strictly prefers $y$ to all other policies and the outcome is Pareto optimal. If the dictator initiates a move to (say) the right, then the mechanism either stops at the ideal policy of some agent, or it stops at a policy that is  to the left of the dictator's ideal point  and to the right of  the ideal point of the agent who chose to stop. In either case, Pareto optimality is satisfied.\footnote{On a side note, if we only demanded for any agent with ideal point $y \in \mathbb{Z}$, that for
every $y''>y'\ge y$ or $y''<y'\le y$, this agent weakly (rather than strictly) prefers $y'$ over $y''$,
then after someone says ``stop'' at some value $y$, to ensure Pareto
optimality, the mechanism would start going in the other direction until some
agent (who was allowed to say stop w.r.t.\ $y$) says stop again, which such a
player does not do as long as she is indifferent between the current value and
the one that will follow it.}

If we require our social choice function to cover
 finer and finer grids in $\mathbb{R}$, then only the above mechanism is OSP-implementable.
However, with our fixed grid, namely $Y=\mathbb{Z}$, the set of obviously strategyproof and unanimous mechanisms is slightly larger than the set of
dictatorships with safeguards against extremism as defined above. We may then combine dictatorships with safeguards against extremism with the
 proto-dictatorships of Theorem
\ref{theorem: two options}. When such a mechanism moves to the right or left, some agents may in addition to stopping or continuing at $y$ call for ``arbitration'' between $y$ and a directly neighboring option.
More specifically, if, e.g., the mechanism goes right from $y'$ to $y'\!+\!1$, some specific agent with $r_i=y'\!+\!1$ may not only force the outcome to be $y'\!+\!1$, but may also (alternatively) choose to initiate an
``arbitration'' between $y'$ and $y'\!+\!1$ via a proto-dictatorship (whose parameters depend on $y'\!+\!1$).\footnote{Similarly, if the set of ideal points is bounded from (say) above by some $\overline{y}$, then one specific agent with
$r_i=\overline{y}$ may choose arbitration between $\overline{y}\!-\!1$ and $\overline{y}$.} In such a case, the obviously strategyproof implementation allows $i$ the
choice between forcing $y'\!+\!1$, initiating an arbitration, and continuing, immediately after all relevant players were given the option to stop at
$y'$ and before any other player is given the option to stop at $y'\!+\!1$.

In an upcoming working paper that was prepared without knowledge of the current paper, \citet{AMN2016} also study OSP-implementable social choice functions on the domain of single peaked preferences, but focus on which coalition can be formed. Recast into the language of our paper, their results show that the possibility of arbitration (which we consider to be a side effect of discretization) can be used to construct coalition systems, and characterize these possible systems. Take, for example, a dictatorship with safeguards against extremism with $r_2=5$, so if the dictator wishes to go right at $4$, then agent $2$ can decide to stop at $5$ but not at $4$. Assume, now, that at $5$ agent $2$ can also choose to initiate an arbitration between $4$ and $5$. Assume, furthermore, that the proto-dictatorship implementing this arbitration is simply a choice by agent $3$ between the outcomes $4$ and $5$. The resulting mechanism is such that agent $2$ can stop unilaterally at $5$, but can force the mechanism stop at $4$ only when joining forces with agent $3$ (indeed, for the outcome $4$ to be implemented, agent $2$ must initiate the arbitration, \emph{and} agent $3$ must choose $4$ in the arbitration), so the coalition of both of these agents is needed to stop at $4$.

\begin{theorem}\label{theorem: GS}
 If $Y=\mathbb{Z}$, a social choice function $\scf$ for the domain of single peaked preferences is unanimous and
OSP-implementable if and only if it is a dictatorship with safeguards against extremism (with the possibility of arbitration as defined above). Any such  $\scf$ is moreover Pareto optimal.
\end{theorem}

The proof of Theorem \ref{theorem: GS} follows from Lemmas \ref{proof 0} through~\ref{proof 5} that are given below, and from Theorem~\ref{theorem: two options}.
The proof of these lemmas, along with the statement and proof of the supporting Lemmas~\ref{proof 1} through~\ref{proof 2.5}, is relegated to Appendix~\ref{app: single peaked}.

\begin{lemma}\label{proof 0}
Any dictatorship with safeguards against extremism is Pareto optimal (and in particular unanimous) and OSP-implementable.
\end{lemma}

The remaining Lemmas~\ref{proof 2} through~\ref{proof 5} apply to any obviously incentive compatible gradual revelation mechanism $M$ for single peaked preferences, under an additional assumption that can be made without loss of generality; see the paragraph opening Appendix~\ref{app: single peaked} for the full details.

\begin{lemma}\label{proof 2}
The set $Y^*_{\emptyset}$ is nonempty. There exist numbers $L^*\in\{-\infty\}\cup\mathbb{Z}$ and $H^*\in\mathbb{Z}\cup\{\infty\}$ s.t.\
$L^*\leq H^*$ and $Y^*_\emptyset=[L^*,H^*]\cap\mathbb{Z}$
\end{lemma}

\begin{lemma}\label{proof 3}
Assume without loss of generality that $P(\emptyset)=1$.
Following are all the actions in~$\overline{A^*_\emptyset}$.
\begin{enumerate}
\item
If $H^*<\infty$, then $\overline{A^*_\emptyset}$ contains an action $r$ with $\mathcal{R}_1(r)=\{R_1: \mbox{ ideal point of }R_1>H^*\}$ and
$Y(r)=[H^*,\infty)\cap\mathbb{Z}$.
\item
If $-\infty<L^*$, then $\overline{A^*_\emptyset}$ contains an action $l$ with $\mathcal{R}_1(l)=\{R_1: \mbox{ ideal point of }R_1<L^*\}$ and
$Y(l)=(-\infty,L^*]\cap\mathbb{Z}$.
\end{enumerate}
\end{lemma}

\begin{lemma}\label{proof 4}
Let $h'$ and $h=(h',a')$ be two consecutive histories of $M$. Assume that $F:=\max Y^*_{h'}\in\mathbb{Z}$ and that $\mathcal{R}_k(h)$ contains (not necessarily exclusively) all preferences with peak $>F$, for every $k\in N$.
If $Y(h) = [F,\infty)\cap\mathbb{Z}$, then (precisely) one of the following holds:
\begin{enumerate}
\item
$Y^*_h=\{F\}$ and $\overline{A^*_h}=\{r\}$ with $Y(h,r)=[F,\infty)\cap\mathbb{Z}$,
\item
$Y^*_h=\{F,F+1\}$ and $\overline{A^*_h}=\{r\}$ with $Y(h,r)=[F+1,\infty)\cap\mathbb{Z}$, or
\item
$Y^*_h=\{F+1\}$ and $\overline{A^*_h}=\{a,r\}$ with $Y(h,a)=\{F,F+1\}$ and $Y(h,r)=[F+1,\infty)\cap\mathbb{Z}$.
\end{enumerate}
\end{lemma}

\noindent The ``mirror version'' of Lemma~\ref{proof 4}) holds for the left: 

\begin{lemma}\label{proof 5}
Let $h'$ and $h=(h',a')$ be two consecutive histories of $M$. Assume that $F:=\min Y^*_{h'}\in\mathbb{Z}$ and that $\mathcal{R}_k(h)$ contains (not necessarily exclusively) all preferences with peak $<F$, for every $k\in N$.
If $Y(h) = (-\infty,F]\cap\mathbb{Z}$, then (precisely) one of the following holds:
\begin{enumerate}
\item
$Y^*_h=\{F\}$ and $\overline{A^*_h}=\{l\}$ with $Y(h,l)=(-\infty,F]\cap\mathbb{Z}$,
\item
$Y^*_h=\{F,F-1\}$ and $\overline{A^*_h}=\{l\}$ with $Y(h,l)=(-\infty,F-1]\cap\mathbb{Z}$, or
\item
$Y^*_h=\{F-1\}$ and $\overline{A^*_h}=\{a,l\}$ with $Y(h,a)=\{F,F-1\}$ and $Y(h,l)=(-\infty,F-1]\cap\mathbb{Z}$,
\end{enumerate}
\end{lemma}

The characterization follows from Lemmas~\ref{proof 2} through~\ref{proof 5}: By Lemma~\ref{proof 3}, if the dictator is not happy with any
option he can force, then he chooses $l$ (left) or $r$ (right), according to where his ideal point lies. Assume w.l.o.g.\ that he chooses to go
right. Then initialize $F=H^*$, and by Lemma~\ref{proof 4} (if he chooses to go left, then Lemma~\ref{proof 5} is used), some other player
is given one of the following three choice sets.
\begin{enumerate}
\item
Action $1$: Force $F$, Action $2$: continue, where $F$ (and everything higher) is still ``on the table''.
\item
Action $1$: Force $F$, Action $2$: force $F\!+\!1$, Action $3$: continue, where only $F\!+\!1$ (and everything higher) is on the table.\\
(So in this case, this agent is the last to be able to stop at $F$ and the first to be able to stop at $F\!+\!1$.)
\item
Action $1$: Force $F\!+\!1$, Action $2$: restrict to $F,F\!+\!1$ (``arbitrate,'' from here must start an onto OSP mechanism that chooses between these two
options, i.e., a proto-dictatorship), Action~$3$: continue, where only $F\!+\!1$ (and everything higher) is on the table.
\end{enumerate}
If this agent chooses continue while keeping $F$ on the table, then some other agent is given one of these three choice sets. If, alternatively, this agent chooses to continue with only $F\!+\!1$ (or higher) on the table, then $F$ is incremented
by one and some other agent is given the one of these three choice sets with the ``new'' $F$.

For any $F \ge H^*$, let $D_F$ be the set of players that were given the option to force $F$ as outcome. ($D_{H^*}$ includes the dictator by
definition.) We claim that $D_F$ is nondecreasing in~$F$. Indeed, for any player who was given the option to force the outcome to be $F$ but not
to force the outcome to be $F\!+\!1$, we have a contradiction w.r.t.\ the preferences that prefer $F\!+\!1$ the most and $F$ second, as by unanimity she
cannot force $F$, but therefore she may end up with $F\!+\!2$ or higher. Finally, note that for any given $F$, only one player can choose to arbitrate
between $F$ and $F\!+\!1$, and since that player can force $F\!+\!1$ at that point, by strategyproofness it follows that she was not given the option to
force $F$ before that, and so the history at which he was allowed to choose ``arbitrate'' was the first history at which he was given any
choice.

\section{Combinatorial Auctions}\label{sec: quasilinear}

In a combinatorial auction, there are $m>0$ goods and $n>1$ agents, called bidders.
In such a setting, an outcome is the allocation of each good to some bidder along with a specification of how much to charge each bidder. 
Each bidder has a nonnegative integer valuation for each bundle of goods, and bidder preferences are represented by utilities that are quasilinear in money: the utility of each bidder from an outcome is her valuation of the subset of the goods that is awarded to her, minus the payment she is charged.
We assume that the possible set of valuations contains (at least) all additive ones: where a bidder simply values a bundle of goods at the sum of her valuations for the separate goods in the bundle.
In this setting, it is customary to define Pareto optimality with respect to the set containing not only all bidders, but also the auctioneer who receives the revenue from the auction. (Otherwise no Pareto optimal outcome exists, as the auctioneer can always pay more and more money to each bidder.) Under this definition, and assuming that goods are worthless to the auctioneer if unsold, Pareto optimality is equivalent to welfare maximization: each good is awarded to a bidder who values it most. Furthermore, when considering combinatorial auctions, it is customary to also require that losers pay nothing.

\citet{Li2015} shows that if $m=1$, then an ascending-price implementation of a second-price auction (which is Pareto optimal and charges losers nothing) is obviously strategyproof. We will now show that this is as far as these properties can be stretched in combinatorial auctions, i.e., that for $m>1$, no social choice function satisfies these properties. In particular, even when all valuations are additive, VCG with the Clarke pivot rule \citep{Vickrey1961,Clarke1971,Groves1973} is not OSP-implementable. (Due to discreteness of the valuation space, there are other Pareto optimal and incentive compatible social choice functions that charge losers nothing beside VCG with the Clarke pivot rule.\footnote{E.g., modifying VCG with the Clarke pivot rule so that any winner who pays a positive amount gets a discount of half a dollar, does not hurt Pareto optimality or strategyproofness (and still charges losers nothing) if all valuations are restricted to be integers.})

\begin{theorem}\label{theorem: quasilinear}
For $m\ge2$ goods, no Pareto optimal (or equivalently, welfare maximizing) social choice function that charges losers nothing is OSP-implementable.
\end{theorem}

It is enough to prove Theorem~\ref{theorem: quasilinear} for $m=2$ goods and $n=2$ bidders, as this is a special case of any case with more goods and/or more bidders.
The proof is by contradiction: we restrict (prune, in the language of \citet{Li2015}) the preference domain to consist of precisely three specifically crafted types $t_1,t_2,t_3$, each corresponding to an additive valuation. We show that even after this restriction, whichever agent who moves first has no obviously dominant strategy. (Regardless of whether VCG or some other social choice function satisfying the above properties is implemented.)
Say that agent~$1$ moves first. Since she has more than one possible action, then one of her actions, say $a$, is chosen by only one of her types $t_i$ (i.e., the other two types do not choose $a$). If agent 1 of type $t_i$ chooses~$a$, then she obtains small (or zero) utility under the worst-case scenario of the other agent turning out to be of the same type. It is then key to craft the three possible types such that for each type $t_i$ (where $i\in \{1,2,3\}$) there is a deviation $t_j$ with $j\ne i$ such that agent~1 of type $t_i$ obtains rather high utility pretending to be $t_j$ when the other agent declares herself to be of the best-case type $t_k$ for which $t_i$'s utility is maximized.
This proof contains elements not found in previous pruning proofs \citep{Li2015,AG2015}, both due to it ruling out a range of social choice functions rather than a single one (therefore working with bounds on, rather than precise quantification of, the utility and payment for every preference profile), and since while all previous such proofs restrict to a domain of preferences of size $2$, this proof restricts to a domain of preferences of size~$3$, which requires a qualitatively more elaborate argument. The full details of the proof are relegated to Appendix~\ref{app: quasilinear}.

\section{House Matching}\label{sec: house-matching}

In a house matching problem, the set of outcomes $Y$ consists of all one-to-one perfect matchings between 
agents in $N$ and houses in a given set $O$ with $|O|\ge|N|$.\footnote{The assumptions of a perfect matching (i.e., that all agents must be matched) and of $|O|\ge|N|$ are for ease of presentation. See Appendix~\ref{homeless-bosses} for a discussion on how our analysis extends if this is not the case.} Each agent only cares about the house she is matched with. 
In this section, we show that a Pareto optimal social choice function for this domain is
OSP-implementable if and only if it is implementable via \textbf{sequential barter with lurkers}. 

To make our definition of sequential barter with lurkers more accessible, while at the same time facilitating the comparison with other results on the OSP-implementability of matching mechanisms, we first define \textbf{sequential barter} --- without lurkers. Sequential barter establishes matchings in trading rounds. In each such round, each agent points to her preferred house. Differently from the trading cycles mechanisms in the literature, 
 houses point to agents gradually. As long as no agent is matched, the mechanism chooses an increasing set of houses and has them point to agents.  
 These choices may be 
 based on the preferences of already-matched agents. At any round, at most two agents are pointed at. Once a cycle forms, the agents and houses in that cycle are matched. Consequently, all houses that pointed to agents matched in this step reenter the process.

With an eye toward the definition of sequential barter \emph{with lurkers}, it is instructive to also consider the following equivalent description of sequential barter (without lurkers). A mechanism is a sequential barter mechanism if and only if it is equivalent to a mechanism of the following form:

\paragraph{Sequential Barter}
\begin{enumerate}
\item
Notation: The sets $O$ and $T$ respectively track the set of all unmatched houses and the set of all active traders.\footnote{An invariant of the mechanism is that $|T|\le2$.} For each active trader $i$, the set $D_i\subseteq O$ tracks the set of houses that $i$ was endowed with (i.e., offered to choose from).
\item
Initialization:
$O$ is initialized to be the set of all houses; $T\gets \emptyset$.\\
So, at the outset all houses are unmatched, and there are no active traders.
\item
Mechanism progress: as long as there are unmatched agents, perform an \textbf{endowment step}.
\begin{itemize}
\item
\textbf{Endowment step}:
\begin{enumerate}
\item
Choose\footnote{All choices in the mechanism may depend on all preferences already revealed.} an unmatched agent $i$, where $i$ must be in $T$ if $|T|=2$.
\item
If $i\notin T$, then initialize $D_i\gets\emptyset$ and update $T\gets T\cup\{i\}$.
\item
Choose some $\emptyset\ne H\subseteq O\setminus D_i$.
\item
Update $D_i\gets D_i\cup H$ and perform a \textbf{question step} for $i$.
\end{enumerate}
\item
\textbf{Question step} for an agent $i\in T$:
\begin{enumerate}
\item
Ask $i$ whether the house she prefers most among $O$ is in $D_i$.  If so, then ask~$i$ which house that is, and perform a \textbf{matching step}  for $i$ and that
house.\\
(If not, then the current mechanism round ends, and a new \textbf{endowment step} is initiated.)
\end{enumerate}
\item
\textbf{Matching step} for an agent $i\in T$ and a house $o$:
\begin{enumerate}
\item
Match $i$ and $o$.
\item
Update $T\gets T\setminus\{i\}$ and $O\gets O\setminus\{o\}$.
\item
$i$ discloses her full preferences to the mechanism.
\item
If $T\ne\emptyset$, then for the unique agent $j\in T$:
\begin{enumerate}
\item
If $o \in D_j$, then set $D_j\gets O$.
\item
Perform a \textbf{question step} for $j$.
\end{enumerate}
(If $T=\emptyset$, then the current mechanism round ends, and a new \textbf{endowment step} is initiated unless there are no more unmatched agents.)
\end{enumerate}
\end{itemize}
\end{enumerate}

All previously known OSP-implementable and Pareto optimal mechanisms for house matching are special cases of sequential barter. 
\citet{Li2015} already shows that the popular (and Pareto optimal) top trading cycles (TTC) mechanism \citep{ShapleyScarf1974} is not obviously strategyproof, yet that serial dictatorship is.  \citet{AG2015} is the first paper to follow-up on \citet{Li2015} and apply obvious strategyproofness. 
Studying marriage problems, they show that 
 no stable matching mechanism is obviously strategyproof for either the men or the women. Due to the overlap between unilateral and bilateral matching theory, the analysis of \citeauthor{AG2015} also  applies to the house matching domain studied here. 
They in particular show that the following generalizations of bipolar serially dictatorial rules \citep{BDE05} can be OSP-implemented: At each mechanism step, either choose an agent and give her free choice among all unmatched houses, or choose two agents, partition all unmatched houses into two sets, and each of the agents gets priority in one of the sets, i.e., gets free pick from that set. If any agent chooses from her set, then the other gets to pick from all remaining houses. If both agents did not choose from their sets, then each gets her favorite choice (which is in the set of the other). \citet{Troyan2016} generalizes even further by showing that any top trading cycles mechanism \cite{AS1999} where at any given point in time no more than two agents are pointed to, is OSP-implementable.  \citet{Troyan2016} also shows that no other TTC mechanism is OSP-implementable.

We now relate our work to \citet{Pycia2016}, which came out several months before our paper, and to \citet{PyciaTroyan2016}, which subsumed \citet{Pycia2016} and came out a couple of days before our paper. \citet{Pycia2016} considers a condition somewhat stronger than OSP, called strong OSP.\footnote{Strong OSP was implicitly assumed in \citet{Pycia2016}; this assumption was made explicit in \citet{PyciaTroyan2016}.} \citet{Pycia2016}
characterizes the sets of matching mechanisms that respectively are strong-OSP imple\-mentable and strong-OSP implementable as well as Pareto optimal as bossy variants of serial dictatorship.
\citet{Pycia2016} uses this result to show that random serial dictatorship is the unique symmetric and efficient rule satisfying strong OSP.
The results added in \citet{PyciaTroyan2016} consider OSP (rather than strong OSP) mechanisms on sufficiently rich preference domains without transfers.\footnote{Their richness condition holds for house matching. However, it does not hold for single peaked preferences with more than two possible outcomes (the setting of Section~\ref{sec: single peaked}), and the assumption of no transfers does not hold for combinatorial auctions (the setting of Section~\ref{sec: quasilinear}).} The first version of that paper claimed that any efficient OSP mechanism under their conditions is equivalent to what we call sequential barter (without lurkers). This result was used to characterize random serial dictatorship
as the unique symmetric and efficient OSP mechanism under their conditions.
Responding to our Theorem~\ref{theorem: houses}, which in particular identifies efficient OSP mechanisms that cannot be represented as sequential barter, \citeauthor{PyciaTroyan2016}'s subsequent versions replace their original claim with a correct, \mbox{nonconstructive} characterization of OSP mechanisms under their conditions. Their proof runs along lines roughly similar to a combination of part of our Theorem~\ref{theorem: revelation} (see Footnote~\ref{fn: gradual-revelation-properties}) and of our Lemma~\ref{lemma: one-undetermined} (this machinery already existed in the first version of \citet{PyciaTroyan2016}). Their correct, updated result implies our Theorem~\ref{theorem: two options} on majority voting and indeed recent versions of their paper explicitly state an equivalent result. Their characterization of random serial dictatorship holds via this updated result.

To see that some Pareto optimal and OSP mechanisms cannot be represented as sequential barter, consider the mechanism represented Fig.~\ref{fig-lurker}, where three traders are active at some history.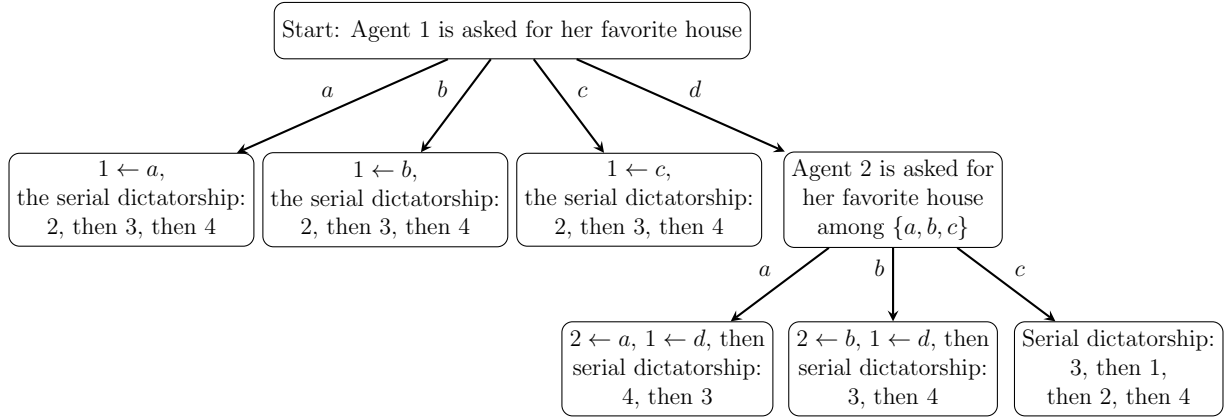
\begin{figure}[ht]
\centering
\begin{tikzpicture}[scale=0.75, every node/.style={transform shape},level distance=30mm]
\tikzstyle{level 1}=[sibling distance=45mm]
\tikzstyle{level 2}=[sibling distance=40mm]
\node [internal] {Start: Agent 1 is asked for her favorite house}
    child {node [internal,align=center] {$1\gets a$,\\the serial dictatorship:\\$2$, then $3$, then $4$}
        edge from parent [arrow] node[above left]{$a$}
    }
    child {node [internal,align=center] {$1\gets b$,\\the serial dictatorship:\\$2$, then $3$, then $4$}
        edge from parent [arrow] node[above left]{$b$}
    }
    child {node [internal,align=center] {$1\gets c$,\\the serial dictatorship:\\$2$, then $3$, then $4$}
        edge from parent [arrow] node[above right]{$c$}
    }
    child [arrow] {node [internal,align=center] {Agent $2$ is asked for\\her favorite house\\among $\{a,b,c\}$}
        child {node [internal,align=center] {$2\gets a$, $1\gets d$, then\\serial dictatorship:\\$4$, then $3$}
          edge from parent [arrow] node[above left]{$a$}
        }
        child {node [internal,align=center] {$2\gets b$, $1\gets d$, then\\serial dictatorship:\\$3$, then $4$}
          edge from parent [arrow] node[above left]{$b$}
        }
        child {node [internal,align=center] {Serial dictatorship:\\$3$, then $1$,\\then $2$, then $4$}
          edge from parent [arrow] node[above right]{$c$}
        }
        edge from parent [arrow] node[above right]{$d$}
    };
\end{tikzpicture}
\caption{An OSP and Pareto optimal mechanism for four agents $1,2,3,4$ and four houses $a,b,c,d$, with three active agents when agent $3$ chooses at the bottom-right.}\label{fig-lurker}
\end{figure}
The (bossy) mechanism in Fig.~\ref{fig-lurker} starts by offering agent $1$ to claim any house among $a$, $b$, and~$c$.
The crux of this mechanism is in that if agent $1$ chooses not to claim any of these houses, then the mechanism can deduce that agent $1$ prefers house $d$ the most, and so at this point a match between agent $2$ and house $d$ can be ruled out without violating Pareto optimality. Even though agent $2$ moves before agent $3$, the competition over house $d$ is now only between agents $1$ and $3$.
This allows the three agents $1$, $2$, and $3$ to be active at the same time. We note that this mechanism had to be crafted in quite a delicate manner to maintain obvious strategyproofness beyond this point: if agent $3$ chooses $d$, then agent $1$ must --- to maintain strategyproofness --- once again be offered to choose between $a$, $b$, and $c$; conversely, if agent $3$ chooses a different house, then since this house may be agent 1's second-most preferred house, to maintain strategyproofness in this case agent 1 must be matched with house $d$.

So, if a mechanism deduces that some agent prefers some house the most (the only way to deduce this without violating OSP and without offering this house to this agent is to offer all other possible houses to this agent; in this case, we say that this agent is a \emph{lurker} for that house), then the mechanism may decide not to allow some other agents to ever get this house, and this allows the introduction of additional traders (beyond two traders) under certain delicate constraints. We are now ready to present our characterization of OSP-implementable and Pareto optimal social choice functions.
A mechanism of sequential barter with lurkers is of the following form:

\paragraph{Sequential Barter with Lurkers}
\begin{enumerate}
\item
Notation: The sets $O$, $T$, $L$, and $G$ respectively track the set of all unmatched houses, the set of all active traders, the set of all lurkers (i.e., all active traders who lurk houses), and the set of all houses who don't have lurkers. For each active trader $i$, the sets $D_i$ and $O_i$ respectively track the set of houses that $i$ was endowed with (i.e., offered to choose from), and the set of houses that $i$ may possibly be matched to.
\item
Initialization:
$O$ is initialized to be the set of all houses; $T\gets \emptyset$,
$L\gets \emptyset$, and
$G\gets O$. So, at the outset all houses are unmatched, there are no active traders (including lurkers), and no house has a lurker.
\item
Mechanism progress: as long as there are unmatched agents, perform an \textbf{endowment step}.
\begin{itemize}
\item
\textbf{Endowment step}:
\begin{enumerate}
\item
Choose\footnote{All choices in the mechanism may depend on all preferences already revealed.} an unmatched agent $i$, where $i$ must be in $T$ if $|T\setminus L|=2$.
\item
If $i\notin T$, then:
\begin{enumerate}
\item
Initialize $D_i\gets\emptyset$.
\item
If $T\setminus L=\{j\}$ for some agent $j$ and $O_j\ne G$, then initialize $O_i\gets G$;\\
otherwise, initialize either $O_i\gets G$ or $O_i\gets O$.
\item
Update $T\gets T\cup\{i\}$.
\end{enumerate}
\item
Choose some $\emptyset\ne H\subseteq O_i\setminus D_i$ such that:\footnote{We constrain the mechanism so that it may only choose an agent $i$ in an endowment step if there is a nonempty set of houses $H$ with which the agent can be endowed (i.e., satisfying these constraints) at that step.}
\begin{itemize}
\item
If $H\setminus G \ne \emptyset$, then $H=O_i\setminus D_i$.
\item
If $O_t\ne G$ for $\{t\}=T\setminus(L\cup\{i\})$, then $H\cap D_t=\emptyset$.
\end{itemize}
\item
Update $D_i\gets D_i\cup H$ and perform a \textbf{question step} for $i$.
\end{enumerate}
\item
\textbf{Question step} for an agent $i\in T$:
\begin{enumerate}
\item
Ask $i$ whether the house she prefers most among $O_i$ is in $D_i$. If so, then ask~$i$ which house that is, and perform a \textbf{matching step}  for $i$ and that
house. If not, and if $i\notin L$, then perform a \textbf{sorting step} for $i$.
\end{enumerate}
\item
\textbf{Matching step} for an agent $i\in T$ and a house $o$:
\begin{enumerate}
\item
Match $i$ and $o$.
\item
Update $T\gets T\setminus\{i\}$ and $O\gets O\setminus\{o\}$, and $O_j\gets O_j\setminus\{o\}$ for every $j\in T$.
\item
If $i\in L$, then update $L\gets L\setminus\{i\}$.
\item
If $o\in G$, then update $G\gets G\setminus\{o\}$.
\item
$i$ discloses her full preferences to the mechanism.
\item
For every agent $j\in T$:\footnote{\label{fn: independent-traversal}The outcome of the mechanism does not depend on the order of traversal of
$T$. This insight is what ensures that the mechanism is OSP. See a discussion below.}
\begin{enumerate}
\item
If $o \in D_j$, then set $D_j\gets O_j$.
\item
Perform a \textbf{question step} for $j$.
\end{enumerate}
(After all \textbf{question steps} triggered by the present \textbf{matching step} are resolved, the current mechanism round ends, and a new \textbf{endowment step} is initiated unless there are no more unmatched agents.)
\end{enumerate}
\item
\textbf{Sorting step} for an agent $i\in T\setminus L$:
\begin{enumerate}
\item
If $|L\setminus T|=2$, then let $j$ be the unique agent $j\in T\setminus(L\cup \{i\})$ .
\item
If $O_i=G$ and $O_i\setminus D_i=\{o\}$ for some house $o$ that does not satisfy $o\in D_j$, then $i$
becomes a \emph{lurker} for $o$:
\begin{enumerate}
\item
Update $L\gets L\cup\{i\}$ and $G\gets G\setminus\{o\}$.
\item
Choose to either keep $O_j$ as is or to
update $O_j\gets O_j\setminus\{o\}$, so that after updating $D_j\subseteq O_j$ holds, and in addition
either $O_j=O$ or $O_j=G$ holds.
\item
If $O_j$ was updated, then perform a \textbf{question step} for $j$.\\
(If not, then the current mechanism round ends, and a new \textbf{endowment step} is initiated unless there are no more unmatched agents.)
\end{enumerate}
\end{enumerate}
\end{itemize}
\end{enumerate}

\begin{remark}
The above mechanism obeys a few invariant properties. $D_i$ is for each active agent $i$ a subset of $O_i$.
For any lurker $i\in L$, the set $O_i\setminus D_i$ contains exactly one house --- the house lurked by $i$, which is preferred by $i$ over every house in $D_i$. At most two active agents are not lurkers at any given time, i.e., $|T\setminus L|\le 2$. There are no lurkers (i.e., $L=\emptyset$) if and only if no unmatched house has a lurker (i.e., $G=O$).
For every $i\in T\setminus L$, either $O_i=G$ or $O_i=O$, and if these options differ (i.e., if $L\ne\emptyset$), then the latter is possible for at most one agent $i\in T\setminus L$.

The sorting step for agent $i$ determines whether $i$ is a lurker --- and should therefore be added to the set~$L$ of lurkers.
This is checked whenever the mechanism infers new information regarding $i$'s preferences, i.e., after each question step for $i$.
Two different types of events, in turn, trigger a question step for $i$: an enlargement of the set $D_i$ and a reduction of the set $O_i$.
The former happens whenever $i$ is offered new houses, i.e., at the conclusion of an endowment step. The latter can happen either due to a house in $O_i$ being matched to another agent, i.e., in a matching step, or due to a house in $O_i$ becoming lurked if $O_i$ is set to (the new) $G$, i.e., in a sorting step.
\end{remark}

Given the introduction of lurkers that precedes the description of the mechanism, the restrictions on the choice of $H$ in the endowment step, and the restrictions on the sets $O_i$ for nonlurkers in the endowment and sorting steps, may seem puzzling. (Essentially, each of these sets $O_i$ is, at each step, either $G$ or $O$, with the former holding for at least one nonlurker.) But, as we will now explain, these restrictions are exactly what drives the obvious incentive compatibility of the mechanism.
\citet{AG2015} have already used in their examples that asking an agent whether she most
prefers some given house, and if the agent's answer is ``yes'' then assigning to her that house (and otherwise continuing), is OSP if the
agent is assured she will eventually get ``at least'' that house. The matching step makes precisely this assurance when allowing
each agent $j$ to be able to claim their top choice from $D_j$ after $O_j$ is reduced, and to claim any house from $O_j$ if a house from $D_j$ becomes matched to another agent.
Therefore, to verify the obvious strategyproofness of the mechanism, what has to be checked is that these assurances, made to several agents in parallel in the matching step, can be simultaneously fulfilled for all these agents. In other words, we have to check that the corresponding loop over agents in the matching step does not depend on the order of traversal of $T$ (see Footnote~\ref{fn: independent-traversal}). 

As it turns out, the above-mentioned ``puzzling'' restrictions on $H$ and on $O_i$ guarantee 
this ``simultaneous fulfillment.'' To see this, envision a scenario with two active nonlurkers $\{i,t\}=T\setminus L$, where $O_i=G$ and $O_t\in\{O,G\}$, and with $\lambda$ lurkers $L=\{1,\ldots,\lambda\}$, such that agent $1$ became lurker first and lurks house $o_1$, agent $2$ became lurker after that and lurks house $o_2$, etc. Note that therefore, $O_1=O$, $O_2=O\setminus\{o_1\}$, $O_3=O\setminus\{o_1,o_2\}$, etc. Assume now that one of these agents $i,t,1,\ldots,\lambda$ chooses a house $o$ in the question step that immediately follows an endowment step.

If the agent that chooses $o$ is $i$, then $o\in O_i=G$, and so, since $o\in D_l$ for each lurker $l$, we have that each lurker $l$ gets free choice from $O_l$, and so each lurker $l$ chooses $o_l$, and there is no conflict (so far) in the choices. If $O_t=G$, then whichever house $t$ prefers most out of $O_t$ (or out of $D_t$) after the removal of $o$ from that set has not been claimed by any lurker, and there can be no conflict between a choice by $t$ and previous choices. On the other hand, if $O_t=O$, then by the first ``puzzling'' restriction on $H$, we have that $D_t\subseteq G$ (indeed, if a house not from $G$ were added at any time to $D_t$ by an endowment step, then by that restriction all houses in $O_t$ were added to $D_t$ and $t$ must have chosen a house in the immediately following question step), and by the second ``puzzling'' restriction on $H$, we have that $o\notin D_t$. Therefore, no house from $D_t$ was claimed by any other agent, and so $t$ is given free choice not from $O_t$ but only from $D_t$ (due to reduction in $O_t$) and there can be no conflict between a choice by $t$ and previous choices.

Now consider the case in which the agent that chooses $o$ is $t$, and that $o\in O\setminus G$ (if $o\in G$, then the previous analysis still holds).
So $o=o_l$ for some lurker $l$. In this case, similarly to the previous analysis, each lurker $l'<l$ has $o_l\in D_{l'}$ and therefore gets free choice from $O_{l'}$ and chooses $o_{l'}$.
So far, all matched houses are $o_1,\ldots,o_l$, so among them only $o_l$ is in $O_l$, and none of them are in $D_l$. So, $o_l$ gets free choice from $D_l$ and there is no conflict (so far) in the choices. If the choice of $l$ is another lurked house $o_{l'}$ (note that $l'>l$), then we reiterate: all lurkers older than $l'$ get their lurked house, and $l'$ gets free choice from $D_{l'}$, etc. This continues until some lurker chooses a house in $G$. Now, as in the previous case (of $i$ choosing $o\in G$), each remaining lurker gets matched to her lurked house with no conflicts. It remains to verify that if $i$ makes a choice, then it does not conflict with any of the choices described so far. This is done precisely as in the case $O_t=G$ of the previous case (of $i$ choosing~$o$): so far, only one matched house was not a lurked house, so only one house was removed from $O_i$; therefore, even if $i$ gets free choice from all remaining houses in $O_i$, there would be no conflict.

Finally, if the agent that chooses $o$ is a lurker $l$,\footnote{We remark that in the analysis below, $l$ in this case is not called a lurker but a dictator. We omit this distinction from the mechanism presentation, as it is not needed for complete and correct presentation of the mechanism, and would only add clutter to it.} then by a similar argument, all lurkers  $l'<l$ get matched to their lurked house. Then, $t$ gets to choose, but recall that since only lurked house were matched so far and since $D_t\subseteq G$ (by the first ``puzzling'' restriction on~$H$), then $t$ chooses from $D_t$, so no conflict arises. If $t$ makes a choice, then the remainder of the analysis is the same as the first case (of $i$ choosing $o\in G$).

So, sequential barter with lurkers is OSP-implementable. Pareto optimality follows from the fact that whenever a set of houses leaves the game, then one of them (the first in the order surveyed in the corresponding explanation in the above three paragraphs) is most-preferred by its matched agent among all not-previously-matched houses, another (the second in the same order) is most-preferred by its matched agent among all not-previously-matched houses except for the first, etc.

\begin{theorem}\label{theorem: houses}
A Pareto-optimal social choice function in a house matching problem is OSP-imple\-mentable if and only if it is implementable via sequential
barter with lurkers.
\end{theorem}

The proof of Theorem \ref{theorem: houses} along with the statement and proof of the supporting Lemmas \ref{lemma: SBL is OSP} through~\ref{lemma:
other-lurker-cant-force-from-terminator}, is relegated to Appendix~\ref{app: house-matching}. The adaptation of the proof to the case of matching with outside options, i.e., where agents may prefer being unmatched over being matched to certain houses (and possibly more agents exist than houses) is described in Appendix~\ref{homeless-bosses}.

\section{Conclusion}\label{sec: conclusion}

This paper characterizes the set of OSP-implementable and Pareto optimal social choice functions in the three most popular domains that allow for strategyproof and Pareto optimal social choice functions that are nondictatorial. We show that obvious strategyproofness rules out many of the most popular mechanisms in these domains, but also gives rise to reasonable mechanisms in some domains, and even to rather exotic and quite intricate mechanisms in other domains.

For single peaked preferences, while median is not obviously strategyproof, some interesting mechanisms are. Dictatorships with safeguards against extremism indeed seem to be reasonable mechanisms: in some policy problems, we may generally delegate decision making to a ``dictator,'' and only if this dictator wishes to adopt some extreme position, should there be some checks in place. It also  seems that such mechanism would be no harder to participate in than ascending auctions. 
For quasilinear preferences, even when all bidders have additive preferences, we complement the elegant positive result of \citet{Li2015} regarding ascending auctions, with a strong impossibility result. To put the restrictiveness of strategyproofness in the setting of auctions into relief, consider two sequential ascending auctions, where only one bidder takes part in both and where the other bidders in the second auction know neither the bids nor the outcome of the first. Even when the preferences of the bidder who takes part in both auctions are additive, bidding her true values for the two goods is not obviously dominant. However, it seems hard to justify that the strategyproofness of such sequential ascending auctions, possibly held months apart, should be any less ``obvious'' (in the colloquial sense of the word) to such an additive bidder than the strategyproofness of a single ascending auction. Finally, for house matching, the mechanisms that we identify are quite complex, and reasoning about them (in fact, even presenting them) seems far from a natural meaning of ``obvious.'' Indeed, while in other known OSP mechanisms, a short argument for the obvious dominance of truthtelling in any history can be written down, in sequential barter with lurkers, not only can the argument for one history be complicated and involve nontrivial bookkeeping, but it can also significantly differ from the (complicated) argument for a different history.

An integrated examination, of all of these negative and positive results, indicates that obvious strategyproofness may not precisely capture the intuitive idea of  ``strategyproofness that is easy to see.'' Indeed, for quasilinear preferences it overshoots in a sense, suggesting that the boundaries of obvious strategyproofness are significantly less far-reaching than one may hope. Conversely, in the context of house matching this definition undershoots in a sense, as it gives rise to some mechanisms that one would not naturally describe as ``easy to understand.'' In this context, we see various mechanics that come into play within obviously strategyproof mechanisms that are considerably richer and more diverse than previously demonstrated.

An interesting question for future research could be to search for an alternative (or slightly modified) concept of easy-to-understand-strategyproofness. One could, for example, consider the similarity, across different histories, of a short argument for the dominance of truthtelling.  A mechanism would then be considered easy to understand if a small set of simple arguments can be used to establish that truthelling is dominant at any possible history. Perhaps such a definition could encompass sequential ascending auctions for additive bidders, while precluding the general form of sequential barter with lurkers?

Regardless of whether OSP catches on or some alternative definition emerges, the fundamental contribution of \citet{Li2015} in moving the discussion of ``strategyproofness that is easy to see'' into the territory of formal definitions and precise analysis will surely linger on.

\bibliographystyle{abbrvnat}
\bibliography{OSP+PO}

\clearpage

\appendix

\section{Proof of Theorem~\ref{theorem: revelation} and Preliminary Analysis}\label{app: revelation}

\begin{proof}[Proof of Theorem~\ref{theorem: revelation}]
Fix any mechanism $M:\mathcal{S}\to Y$, social choice function $\scf:\mathcal{R}\to Y$ and obviously strategyproof strategy profile
$\mathbf{S}:\mathcal{R}\to
\mathcal{S}$. Define a new mechanism $M^1:\mathcal{S}^1\to Y$ that is identical to $M$ except that all information sets are singletons. In $M$,
a
strategy for agent $i$ is an
 $\mathcal{I}_i$-measurable function mapping a nonterminal history $h\in P^{-1}(i)$ of player $i$ to the actions $A(h)$ available at $h$. In
 $M^1$, a strategy for agent $i$ has the same definition, only without the requirement of being $\mathcal{I}_i$-measurable. So, we have
 $\mathcal{S}\subseteq\mathcal{S}^1$. Since $\mathbf{S}$ is obviously strategyproof in $M$, we obtain that $\mathbf{S}$ is also obviously
 strategyproof in $M^1$.\footnote{Since $\mathcal{S}_i$ is obviously dominant in $M$, $M(\mathcal{S}_i(R_i),B_{-i})R_iM(B')$
 holds for all  $R_i\in \mathcal{R}_i$,  $B$, $B'$,  $h$, and $h'$ with $h\in\Path(\mathcal{S}_i(R_i),B_{-i})$, $h'\in\Path(B')$, $P(h)=P(h')=i$, $\mathcal{I}_i(h)=\mathcal{I}_i(h')$, and $\mathcal{S}_i(R_i)(h)\neq B'_i(h')$. So $M(\mathcal{S}_i(R_i),B_{-i})R_iM(B')$ in particular holds whenever $h=h'$ in the above conditions.}

For every history $h$ in $M^1$ and $i$, set $\mathcal{R}_i(h)=\{R_i\in\mathcal{R}_i\mid \exists B_{-i} : h\in
\Path(\mathbf{S_i}(R_i),B_{-i})\}$.
If $P(h)=i$, then we note that $\{\mathcal{R}_i(h,a)\mid a\in A(h)\}$ partitions $\mathcal{R}_i(h)$, however some of the sets in this partition may be empty.
Since $\mathbf{S}$ is obviously strategyproof in $M^1$, $\mathbf{T}$ is obviously strategyproof in $M^1$ as well (w.r.t.\ the maps just
defined). Furthermore, $M^1(\mathbf{T}(\cdot))=M^1(\mathbf{S}(\cdot))=\scf(\cdot)$.

In the following steps of the proof, we describe modifications to the game tree of the mechanism. For ease of presentation, we consider the maps
$P(\cdot)$, $S_i(R_i)(\cdot)$ and $\mathcal{R}_i(\cdot)$ to be defined over nodes of the tree (rather than histories, as we may modify the paths to these nodes).

Let $M^2_0=M^1$. For every agent $i\in N$ (inductively on $i$), we define a new mechanism $M^2_i:\mathcal{S}^2_i\to Y$ as follows: For every
preference
$R_i$, for every minimal history $h$ s.t.\ $P(h)=i$ and $\{M^2_{i-1}(\mathbf{T}(R))\mid h\in\Path(\mathbf{T}(R))\}$ is a nonempty
set of completely $i$-indifferent outcomes,
let $a=\mathbf{T}_i(R_i)(h)$, remove $R_i$ from the set $\mathcal{R}_i(h,a)$, and put $\mathcal{R}_i(h,a')=\{R_i\}$ for a new action $a'$ at $h$
that
leads to a subtree that is a duplicate of that to which $a$ leads before this change (with all maps from the nodes of the duplicate subtree
defined as on the original subtree). Note that $M^2_i(\mathbf{T}(R))=M^2_{i-1}(\mathbf{T}(R))$ holds for all $R$, so we have
$M^2_i(\mathbf{T}(\cdot))=M^2_{i-1}(\mathbf{T}(\cdot))=\scf(\cdot)$. Moreover, since $\mathbf{T}$ is obviously strategyproof in $M^2_{i-1}$,
$\mathbf{T}$ is obviously strategyproof in $M^2_i$. Set $M^2=M^2_n$.

 Define a new mechanism $M^3:\mathcal{S}^3\to Y$ by dropping from $M^2$ any action $a$ for which there exists no $R$ such that $a$ is on the
 path
 $\Path(\mathbf{T}(R))$ in $M^2$.
Since $\mathbf{T}$ is obviously strategyproof in $M^2$, $\mathbf{T}$ is also obviously strategyproof in $M^3$. Furthermore,
 $M^3(\mathbf{T}(\cdot))=M^2(\mathbf{T}(\cdot))=\scf(\cdot)$.

 Define a new mechanism $M^4:\mathcal{S}^4\to Y$ as follows. Identify a maximal set of histories $H^*$ in $M^3$ that satisfies all of the
 following:
\begin{itemize}
\item
Each $h\in H^*$ is either nonterminal or infinite.
\item
$P(h)=i$ for all nonterminal $h\in H^*$ and some $i$,
\item
there exists a history $h^*\in H^*$ such that $h^*\subseteq h$ for all $h\in H$, and
\item
if $h\in H^*$ then $h'\in H^*$ for all $h'$  with $h^*\subset h'\subset h$.
\end{itemize}
``Condense'' each such $H^*$ by replacing the set of actions $A(h^*)$ at $h^*$ s.t.\ at $h^*$, agent $i$ directly chooses among all possible
nodes $(h,a)$,
where $h$ is a maximal nonterminal history in $H^*$ and $a\in A(h)$; in addition, for every infinite history $h$ in $H^*$, add an action to
$A(h^*)$ that chooses a new leaf with the same outcome as  $h$. For every new action $a'$ that chooses a node $(h,a)$ from $M^3$, set
$\mathcal{R}_i(h^*,a')=\mathcal{R}_i(h,a)$;
for every new action $a'$ that chooses a new leaf with the same outcome as in an infinite history $h=(a^k)_{k=1,\ldots}$ of $M^3$, set
$\mathcal{R}_i(h^*,a')=\cap_k \mathcal{R}_i(a^1,\ldots,a^k)$.
Since $\mathbf{T}$ is obviously strategyproof in $M^3$, $\mathbf{T}$ is also obviously strategyproof in $M^4$. Furthermore,
 $M^4(\mathbf{T}(\cdot))=M^3(\mathbf{T}(\cdot))=\scf(\cdot)$.

Define a new mechanism $M^5:\mathcal{S}^5\to Y$ as follows. Identify a maximal set of histories $H^*$ in $M^4$ that satisfies all of the
following:
\begin{itemize}
\item
$|A(h)|=1$ for all nonterminal $h\in H^*$,
\item
there exists a history $h^*\in H^*$ such that $h^*\subseteq h$ for all $h\in H$, and
\item
if $h\in H^*$ then $h'\in H^*$ for all $h'$  with $h^*\subset h'\subset h$.
\end{itemize}
``Condense'' each such $H^*$ by replacing the subtree rooted at the node $h^*$ with the subtree rooted at the node $h$,
where $h$ is the maximal history in $H^*$. If $h$ is infinite, then replace $h^*$ with a new leaf with the same outcome as $h$ and the same value of the maps $\mathcal{R}_i(\cdot)$ as $h^*$.
Since $\mathbf{T}$ is obviously strategyproof in $M^4$, $\mathbf{T}$ is also obviously strategyproof in $M^5$. Furthermore,
 $M^5(\mathbf{T}(\cdot))=M^4(\mathbf{T}(\cdot))=\scf(\cdot)$.

By construction, $M^5$ is a gradual revelation mechanism that implements $\scf$.
\end{proof}

\begin{lemma}\label{lemma: owned choices}
Fix an obviously incentive compatible gradual revelation mechanism $M$. Let $h$ be a nonterminal history and let $i=P(h)$.
If there exists $y\in Y(h)$ s.t.\ $[y]_i\cap Y(h,a)\ne\emptyset\ne [y]_i\cap Y(h,a')$ for two distinct $a,a'\in A(h)$, and furthermore there
exists $R_i\in\mathcal{R}_i(h)$ s.t.\ $R_i$ ranks $[y]_i$ at
the top among $Y(h)$, then $[y]_i\cap Y(h)\subseteq  Y^*_h$.
\end{lemma}

\begin{proof}[Proof of Lemma~\ref{lemma: owned choices}]
Suppose not, and assume w.l.o.g.\ that $y \notin Y^*_h$.
Since $y\notin Y^*_h$, there exists a preference profile $R_{-i}\in\mathcal{R}_{-i}(h)$ (recall that this is equivalent to $h\in\Path(\mathbf{T}(R))$) such that
$M(\mathbf{T}(R))=y'\notin [y]_i$.
Assume w.l.o.g.\ that $\mathbf{T}_i(R_i)(h)\neq a'$. Since $[y]_i\cap Y(h,a')\ne\emptyset$, there exists a preference profile $R'\in\mathcal{R}(h)$ (recall that this is equivalent to $h\in
\Path(\mathbf{T}(R'))$) with $\mathbf{T}(R')(h)=a'$ such that
$M(\mathbf{T}(R'))\in[y]_i$. Since $M(\mathbf{T}(R'))R_i y P_i y'=M(\mathbf{T}(R))$, the mechanism is not obviously strategyproof, reaching a
contradiction.
\end{proof}

\section{Proof of Theorem~\ref{theorem: two options}}\label{app: voting}

\begin{proof}[Proof of Theorem~\ref{theorem: two options}]
Fix any social choice function $\scf$ that is implementable via an obviously strategyproof mechanism. By Theorem \ref{theorem: revelation},
$\scf$ must be implementable by an obviously incentive compatible gradual revelation mechanism $M$. Let $h,i$ be such that $h$ is a minimal history with $P(h)=i$. Since $M$ is a gradual revelation mechanism, $i$ must have at least two choices at $h$ (i.e., $|A(h)|\geq 2$). Since there are only two possible preferences for $i$ and since $M$ is a gradual revelation mechanism, there are at most $2=|\mathcal{R}_i|$ choices for $i$ at $h$. In sum, we have $|A(h)|=2$. Moreover, there exists no $h'$ with $h\subsetneq h'$ and $P(h')=i$, since $i$ already fully reveals his preference at $h$. By Lemma \ref{lemma: owned choices}, $Y^*_h\neq \emptyset$. So, $h$ must be covered by one of the three above cases.

To see that any proto-dictatorship  is obviously strategyproof, it is enough to analyze histories $h$ in which $Y^*_h=\{y\}$ and
$\overline{A^*_h}=\{\tilde{a}\}$ with $Y(\tilde{a})=\{y,z\}$ (histories $h$ with $Y^*_h=\{z\}$ are analyzed analogously, and in histories $h$ with
$Y^*_h=\{y,z\}$, the choosing agent is a dictator). In this case, $P(h)$ ensures that the outcome is $y$ if $y$ is his preferred option. If $z$
is his preferred option, then choosing $\tilde{a}$ is obviously strategyproof: the best outcome under the
deviation to ensuring $y$ is identical to the worst outcome given $\tilde{a}$.
\end{proof}

\section{Proof of Lemmas~\ref{proof 0} through~\ref{proof 5}}\label{app: single peaked}

\begin{proof}[Proof of Lemma~\ref{proof 0}]
As outlined in Section~\ref{sec: single peaked}.
\end{proof}

For Lemmas \ref{proof 1} through \ref{proof 2.5} and \ref{proof 2} through \ref{proof 5}, fix an obviously incentive compatible gradual revelation mechanism $M$ that implements a given
social choice function, with the following property: For each nonterminal history $h$ of $M$, there does not exist another obviously incentive
compatible gradual revelation mechanism $M'$ that implements the same social choice function as $M$ and such that $M$ and $M'$ coincide except
for the subtree at $h$, and such that $|\overline{A^*_h}|$ is at most $1$ at $M$, but greater than $1$ at $M'$. Such an $M$ always exists: start
with any obviously incentive compatible gradual revelation mechanism $M'$ that implements the given social choice function, and then,
considering first each nodes $h$ in the first level in the tree of $M'$,
 if $h$ violates the above condition, replace the subtree at $h$ with another
subtree that satisfies the above condition. Next replace all subtrees that violate the above condition at nodes in the second level, then in the third, and so forth.  Since each node does not change any more after some finite number of steps (equal to the level of this node), the resulting mechanism is well defined, even though the height of the tree of $M$ (and so the number of steps in the process defining $M$) may be infinite.

Fix an obviously incentive compatible gradual revelation mechanism $M$ that implements a unanimous social choice function. Assume that $M$ satisfies the above property and assume w.l.o.g.\ that $P(\emptyset)=1$.

\begin{lemma}\label{proof 1}
Let $h$ be a nonterminal history in $M$ and let $i=P(h)$.
Let $y\in Y$ s.t.\ $y\in Y(h,a)$ and $y+1\in Y(h,a')$ for two distinct $a,a'\in A(h)$.
If there exist\footnote{Similarly to notation of other sections, we use, e.g., $R_i:y,y+1$ to denote a preference $R_i$ for agent~$i$ that ranks $y$ first and $y+1$ second.} $R_i: y, y+1\in\mathcal{R}_i(h)$ and $R'_i: y+1\in\mathcal{R}_i(h)$  (or  $R_i: y\in\mathcal{R}_i(h)$ and $R'_i: y+1,y\in\mathcal{R}_i(h)$), then $\{y,y+1\}\cap Y^*_{h}\neq \emptyset$.
\end{lemma}

\begin{proof}[Proof of Lemma~\ref{proof 1}]
We prove the lemma for the first case (swapping $y$ and $y+1$ obtains the proof for the second case). Suppose that $y, y+1\notin Y^*_h$.
Assume w.l.o.g.\ that $\mathbf{T}_i(R_i)(h)=a$.
By Lemma~\ref{lemma: owned choices} and by definition of $R_i,R'_i\in \mathcal{R}_i(h)$, we have that $y,y+1\notin Y(h,a)\cap Y(h,a')$, so $y\notin Y(h,a')$
and
$y+1\notin Y(h,a)$. Since $y\notin Y^*_h$, there is some
preference
profile $R_{-i}\in\mathcal{R}_{-i}(h)$ with $M(\mathbf{T}(R))=y'\neq y$. Since the second
ranked choice under $R_i$, namely $y+1$, is not in $Y(h,a)$, we have $y'\neq y+1$. Since $y+1\in Y(h,a')$, there exists a preference profile
$R'\in\mathcal{R}(h)$ with $\mathbf{T}_i(R'_i)(h)=a'$ and \ $M(\mathbf{T}(R))=y+1$. A contradiction to the obvious strategyproofness arises, since $y+1=M(\mathbf{T}(R'))P_i
M(\mathbf{T}(R))=y'$.
\end{proof}

\begin{lemma}\label{proof 1.5}
Let $i\in N$ and let $h$ be a minimal nonterminal history  in $M$ s.t.\ $P(h)=i$. If $Y(h)=Y$, then $Y^*_h\ne\emptyset$.
\end{lemma}

\begin{proof}[Proof of Lemma~\ref{proof 1.5}]
Suppose $Y^*_{h}=\emptyset$.
Since $M$ is a gradual revelation mechanism, $A(\emptyset)$ must contain at least two choices. So, there exists $y\in Y$ such that $y\in
Y(h,a)$, $y+1\in Y(h,a')$ for $a\neq a'$ and $y,y+1\notin Y^*_{h}$ (and recall that $\mathcal{R}_i(h)=\mathcal{R}_i$), a contradiction to Lemma
\ref{proof 1}. So $Y^*_{h}$ must be nonempty.
\end{proof}

\begin{proof}[Proof of Lemma~\ref{proof 2}]
By Lemma~\ref{proof 1.5}, $Y^*_{\emptyset}\ne\emptyset$.
Let $y^*<y^{\circ}$ be two policies in $Y^*_{\emptyset}$. Suppose we had  $y'\in (y^*,y^{\circ})\cap Y$ but $y'\notin Y^*_{\emptyset}$. Since $y'\notin
Y^*_{\emptyset}$, there exists a preference profile $R$ such that $R_1$ ranks $y'$ at the top but the outcome of the mechanism is
$M(\mathbf{T}(R))=\tilde{y}\neq y'$. Assume without loss of generality that $\tilde{y}<y'$.

Define two preference profiles $R'$ and $R''_{-1}$ such that $R'_i:\tilde{y}$ and $R''_i: y'$ for all $i\neq 1$, and such that $R'_1: y'$ and $R'_1$ ranks $y^{\circ}$ strictly above $\tilde{y}$. Starting with the profile $R$ and inductively switching the preference of each agent $i\ne1$ from $R_i$ to $R'_i$, the strategyproofness of $M$ implies that $M(\mathbf{T}(R_1,R'_{-1}))=\tilde{y}$. Since $R'_1$ ranks $y^{\circ}$ strictly above $\tilde{y}$, since $y^{\circ}\in Y^*_{\emptyset}$, since $R'_1$ is single peaked, and since $M$ is strategyproof, we have that $M(\mathbf{T}(R'))\in (\tilde{y},y^{\circ}]$.

Assume for contradiction that $M(\mathbf{T}(R')) > y'$. Since $\tilde{y}<y'$, since $R'_2:\tilde{y}$, and since by strategyproofness for we have (similarly to he above argument for $R'$) that $M(\mathbf{T}(R''_2, R'_{-2}))>\tilde{y}$, we have that $M(\mathbf{T}(R''_2, R'_{-2}))\ge M(\mathbf{T}(R'))$ holds by strategyproofness. (Indeed, if we had $M(\mathbf{T}(R'))>M(\mathbf{T}(R''_2, R'_{-2}))>\tilde{y}$ then agent 2 with preference $R'_2$ would have an incentive to lie.) Inductively switching the preference of each agent $i\ne 1$ from $R'_i$ to $R''_i$ and applying the preceding argument, we obtain that $M(\mathbf{T}(R'_1, R''_{-1}))\ge M(\mathbf{T}(R'))$. Since $M(\mathbf{T}(R'))>y'$ we obtain a contradiction to unanimity, which requires $M(\mathbf{T}(R'_1, R''_{-1}))=y'$ as all preferences in $(R'_1, R'_{-1})$ have the ideal point $y'$.
 Therefore,
$M(\mathbf{T}(R')) \le y'$.

Therefore,
$M(\mathbf{T}(R_1,R'_{-1}))=\tilde{y}<M(\mathbf{T}(R'))\leq y'$, and so $M(\mathbf{T}(R'))P_1M(\mathbf{T}(R_1,R'_{-1}))$, contradicting
the strategyproofness of $M$. So we must have $y'\in Y^*_{\emptyset}$, and therefore $y'\in Y^*_{\emptyset}$ is a nonempty ``interval''.
\end{proof}

\begin{lemma}\label{proof 2.5}
Let $h$ be a nonterminal history. If $\overline{A^*_h}=\{a\}$ for some action $a$, then
for every $y^*\in Y^*_h$, either $y^*\le Y(h,a)$ or $y^*\ge Y(h,a)$.
\end{lemma}

\begin{proof}[Proof of Lemma~\ref{proof 2.5}]
Let $i=P(h)$.
Assume for contradiction that $X = \{y \in Y(h,a) \mid y < y^*\}$ and $Z = \{y \in Y(h,a) \mid y > y^*\}$ are both nonempty for some $y^*\in Y^*_h$. Let

\[Y_X = \{R_i \in \mathcal{R}_i(h) \mid \mathbf{T}_i(R_i)(h)=a \And \mbox{the peak of $R_i$ is $\le y^*$}\},\]
\[Y_Z = \{R_i \in \mathcal{R}_i(h) \mid \mathbf{T}_i(R_i)(h)=a \And \mbox{the peak of $R_i$ is $> y^*$}\}.\]

We claim that for every $R_i \in Y_X$, for every $R_{-i}\in\mathcal{R}_{-i}(h)$ it is the case that $M(\mathbf{T}(R)) \in X \cup
\{y^*\}$. Indeed, if $M(\mathbf{T}(R)) \in Z$, then we would have a contradiction to strategyproofness because at $h$ agent $i$ can ensure that
the outcome is $y^*$, which he prefers over $M(\mathbf{T}(R))>y^*$. Similarly, for every $R_i \in Y_Z$, for every $R_{-i}\in\mathcal{R}_{-i}(h)$ it is the case that $M(\mathbf{T}(R)) \in Z \cup \{y^*\}$.
Since $X$ and $Z$ are nonempty, $Y_X$ and $Y_Z$ are nonempty as well. (Otherwise we would have a contradiction to the
fact that $Y(h,a)$ contains the disjoint union of $X$ and $Y$, and possibly also $y^*$.) 

Define two new two new distinct actions $l$ and $r$. Let  $\mathcal{R}_i(l)=Y_X$ and let $l$ lead to a copy of the subtree that $a$ leads to, where all choices outside $\mathcal{R}_i(l)$ have been pruned. Similarly Let  $\mathcal{R}_i(r)=Y_Z$ and let $r$ lead to a copy of the subtree that $a$ leads to, where all choices outside $\mathcal{R}_i(r)$ have been pruned. Define a new tree by replacing $a\in A(h)$ with the two new actions $l$ and $r$. If this new mechanism $\hat{M}$ is obviously incentive compatible (under the above-defined truthtelling strategy), then it implements the same social choice function as the original mechanism, and so our assumption that if $|\overline{A^*_h}|=1$ then no
modification can result in $|\overline{A^*_h}|>1$ is violated. (Note that indeed we have in $\hat{M}$ that $l,r\in\overline{A^*_h}$ and so $|\overline{A^*_h}|=2$, since by gradual revelation any preference $R_i$ that chooses $a$ in the original mechanism does not ensure the final outcome, and so also every preference $R_i$ that chooses $l$ or $r$ in $\hat{M}$ does not ensure the final outcome.)

To see that the new mechanism $\hat{M} $ is indeed obviously incentive compatible (under the above-defined truthtelling strategy), note that for every history $h'$ that is not a subhistory
of $h$, the conditions for obvious strategyproofness do not change, while for every history $h'$ that is a strict subhistory of $h$, when a
player makes a choice, he faces a pruned version of the subtree he faced in the original mechanism, and therefore obvious strategyproofness is maintained. It
remains to show that obvious strategyproofness is not violated for $i$ at~$h$. Indeed, obvious strategyproofness when playing any action in
$A^*_h$ is maintained since $i$ would --- under the original mechanism and under $\hat{M}$ --- for any deviation at best get his best choice from $Y(h)$, which has not changed.  
 Similarly, obvious strategyproofness when playing $l$ or~$r$ is
maintained w.r.t.\ deviating to forcing some outcome, because the worst outcome when playing $l$ or $r$ in $\hat{M}$ is no worse than the worst outcome when playing $a$ in
the original mechanism, and deviation was not incentivized there. Finally, obvious strategyproofness when playing $l$ is maintained w.r.t.\
deviating to playing $r$, since the best outcome that a type that chooses $l$ (and by definition has peak $\le y^*$) can get by playing $r$ is no
better than $y^*$ (since all possible outcomes when playing $r$ are in $Z\cup\{y^*\}$ and therefore are $\ge y^*$), however the worst outcome
such a type gets when playing $l$ is $\le y^*$.
Similarly, obvious strategyproofness when playing $r$ is maintained w.r.t.\ deviating to playing $l$ as well.
\end{proof}

\begin{proof}[Proof of Lemma~\ref{proof 3}]
Assume that $H^*<\infty$ (the proof when $H^*=\infty$ and $L^*>-\infty$ is analogous to the proof when $H^*<\infty$ and $L^*=-\infty$). Assume for contradiction that there exist two different actions $a,a'\in A(\emptyset)$ such that $Y(a)\cap
{[H^*+1,\infty)}\neq\emptyset \neq Y(a')\cap {[H^*+1,\infty)}$.
So there must exist some $y>H^*$ such that
$y\in Y(a)$ and $y+1\in Y(a')$, a contradiction to Lemma \ref{proof 1}. Therefore there is at most one action $r\in A(\emptyset)$ such that
$Y(r)\cap {[H^*+1,\infty)}\neq \emptyset$. By unanimity, $\bigcup_{a\in A(\emptyset)}Y(a)=Y$, and so for each  $y\in  {[H^*+1,\infty)}$ there must exist an action $a$ such that $y\in Y(a)$. By the preceding two statements,
there exists a unique action $r \in A(\emptyset)$ s.t.\
$Y(r)\cap {[H^*+1,\infty)}\neq \emptyset$, and furthermore, $Y(r)\supseteq{[H^*+1,\infty)}\cap\mathbb{Z}$.

Assume for contradiction that $H^*\notin Y(r)$. Consider the preference profile $R_1$ that ranks $H^*+1$ first, and ranks $H^*$ second. By
unanimity, $\mathbf{T}_1(R_1)(\emptyset)=r$ (as no other action has $H^*+1$ as a possible outcome), however, since $H^*+1\notin Y^*_\emptyset$,
we have that there exists a preference profile $R_{-i}\in\mathcal{R}_{-i}$ s.t.\ $M(\mathbf{T}(R))\ne H^*+1$. Therefore, since
$H^*\in Y^*_\emptyset$, we obtain a contradiction to obvious strategyproofness. Therefore, $Y(r)\supseteq{[H^*,\infty)}\cap\mathbb{Z}$.

Mutatis mutandis if $L^*>-\infty$.
there exists a unique action $l \in A(\emptyset)$ s.t.\ $Y(l)\cap {(-\infty,L^*-1]}\neq \emptyset$, and furthermore,
$Y(l)\supseteq{(-\infty,L^*]}\cap\mathbb{Z}$.
We note that we have not yet determined whether or not $l=r$.

To complete the proof, we reason by cases. Assume first that either $L^*=-\infty$ or $r\ne l$.
Therefore, if $L^*>-\infty$, then $Y(r)\cap{(-\infty,L^*-1]}=\emptyset$, and so $Y(r)\subseteq [L^*,\infty)\cap\mathbb{Z}$.
As regardless of whether or not $L^*>-\infty$, we have in the current case that $Y(r)\subseteq [L^*,\infty)\cap\mathbb{Z}$, to complete the proof
for this case,
it is enough to show that $Y(r)\cap[L^*,H^*-1]=\emptyset$. Assume for contradiction that there exists $y\in Y(r)\cap[L^*,H^*-1]$. Therefore,
there exists a preference profile $R$ s.t.\
$\mathbf{T}_1(R_1)(\emptyset)=r$ and $M(\mathbf{T}(R))=y$.
Since the mechanism is a gradual revelation mechanism, there exists $y'\in Y(r)\setminus\{y\}$
and a preference profile $R'_{-1}$ s.t.\ $M(\mathbf{T}(R_1,R'_{-1}))=y'$. By strategyproofness and since $y\in Y^*_{\emptyset}$, we have that $y'
R_1 y$.
If $y R_1 y'$ as well, then the peak $y''$ of $R_1$ is between $y$ and $y'$; by strategyproofness, therefore $y''\in Y(r) \setminus
Y^*_{\emptyset}$ and so $y''>H^*$.
If $y' P_1 y$, then by strategyproofness, $y'\notin Y^*_{\emptyset}$; therefore, as $y'\in Y(r)$, we have that $y'>H^*$. Either way, there exists
$y''>H^*$ that $R_1$ ranks strictly above $y$. Since $H^*>y$, we therefore obtain a contradiction to strategyproofness, as $H^*\in
Y^*_{\emptyset}$ is preferred by $R_1$ over $M(\mathbf{T}(R))=y$.

It remains to consider the case in which $L^*>-\infty$ (recall also that $H^*<\infty$) and $r=l$, but such a setting is impossible, as it
contradicts Lemma~\ref{proof 2.5}.
\end{proof}

\begin{proof}[Proof of Lemma~\ref{proof 4}]
Let $j=P(h')$ and  $i=P(h)$.
To see that there is no
action $a'$ at $h$ with $Y(h,a')\subseteq[F\!+\!2,\infty)$, suppose  there was such an action. Consider a profile of preferences $R$ with
$R_j: F\!+\!1, F$, with $R_k:F\!+\!1$ for all $k\in N\setminus\{i,j\}$, and with $R_i\in\mathcal{R}_i(h,a')$. (Since for each $k\ne i$, agent $k$'s ideal point is larger than $F$, we have that $R$ reaches $h$.) Since $Y(h,a')\subseteq[F\!+\!2,\infty)$, necessarily $M(\mathbf{T}(R))\geq F\!+\!2$. This is a contradiction to strategyproofness, since $j$ strictly prefers $F$ to $F\!+\!2$ and $M(\mathbf{T}(R'_j,R_{-j}))=F$ for $R'_j:F$ (since $F\in Y^*_{h'}$). There is in sum no
action $a'$ at $h$ with $Y(h,a')\subseteq[F\!+\!2,\infty)$.

Since $Y(h)=[F,\infty)\cap\mathbb{Z}$ and since there is more than one action in $A(h)$ (by gradual revelation),
Lemma~\ref{proof 1} yields $Y^*_h\ne\emptyset$.
If $Y^*_h$ did intersect
$[F\!+\!2,\infty)$, there would be an action $a'\in A(h)$ with $Y(h,a')\subseteq[F\!+\!2,\infty)$, a contradiction to the preceding paragraph.
Therefore, $\emptyset\ne Y^*_h\subseteq\{F,F\!+\!1\}$.

Let $G:=\max Y^*_h$.
Since $Y(h)=\cup_{a'\in A(h)}Y(h,a')$, there exists some action $a'\in A(h)$ such that $[G\!+\!1,\infty)\cap Y(h,a')\ne\emptyset$. By Lemma ~\ref{proof 1}, there is at most one such action, call it $r$. In sum we have $[G\!+\!1,\infty)\cap \mathbb{Z}\subseteq Y(h,r)$. 
To see that $G\in Y(h,r)$, suppose not.
Let $R_i:G\!+\!1,G$. ($R_i\in\mathcal{R}_i(h)$ since agent $i$'s ideal point is larger than $F$.)
We note that truthtelling for agent $i$ requires choosing $r$ at $h$. Since $G\!+\!1\notin Y^*_h$, there exists $R_{-i}\in\mathcal{R}_{-i}(h)$ s.t.\ $M(\mathbf{T}(R))\ne G\!+\!1$.
Since $G\notin Y(h,r)$, we furthermore have that $M(\mathbf{T}(R))\ne G$.
Therefore, $i$ strictly prefers $G$ to $M(\mathbf{T}(R))$, even though $G\in Y^*_h$ --- a contradiction to strategyproofness.
 In sum we have that $[G,\infty)\cap \mathbb{Z}\subseteq Y(h,r)$, and that for every $a'\in A(h)\setminus\{r\}$, the set $Y(h,a')$ does not intersect $[G\!+\!1,\infty)$. We complete the proof by reasoning over cases.

If $Y^*_h=\{F\}$, then since $G=F$ we have that $Y(h,r)\supseteq[F,\infty)\cap\mathbb{Z}$ and so $Y(h,r)=[F,\infty)\cap\mathbb{Z}$.
Since for any action $a'\neq r$, the set $Y(h,a')$ does not intersect $[G\!+\!1,\infty)=[F\!+\!1,\infty)$, we therefore have that $\overline{A^*_h}=\{r\}$, as required.

The two remaining cases are those where $Y^*_h$ is either $\{F,F\!+\!1\}$ or $\{F\!+\!1\}$. In either of these cases, since $G=F\!+\!1$, we have that $Y(h,r)\supseteq[F\!+\!1,\infty)\cap\mathbb{Z}$
and that for any action $a'\ne r$, the set $Y(h,a')$ does not intersect $[F\!+\!2,\infty)$, i.e., $Y(h,a')\subseteq\{F,F\!+\!1\}$.

Assume first that $Y^*_h=\{F,F\!+\!1\}$.
Since any action $a'\ne r$ at $h$ has $Y(h,a) \subseteq \{F,F\!+\!1\} = Y^*_h$, by gradual revelation and strategyproofness, $a' \in A^*_h$. Therefore, $\overline{A^*_h}=\{r\}$. It remains to show that $Y(h,r)=[F\!+\!1,\infty)\cap\mathbb{Z}$, i.e., that $F\notin Y(h,r)$.
Note that for preferences with ideal point $\le F\!+\!1$, agent $i$ can force at $h$ his most-preferred option from $Y(h)$, and so, by gradual revelation and strategyproofness, chooses an action in $A^*_h$, i.e., an action other than~$r$. Therefore, agent $i$ chooses $r$ only if her ideal point is greater than $F\!+\!1$. Assume for contradiction that $F \in Y(h,r)$. Therefore, there exists some
preference profile $R\in\mathcal{R}(h,r)$ with $M(\mathbf{T}(R))=F$. Since $R_i$ chooses $r$ at $h$, the ideal point of $R_i$ is greater than $F\!+\!1$, so $R_i$ strictly prefers $F\!+\!1$ (which he can force at $h$) to $F=M(\mathbf{T}(R))$ --- a contradiction to strategyproofness. So $Y(h,r)=[F\!+\!1,\infty)\cap\mathbb{Z}$, as required.

Finally, assume that $Y^*_h=\{F\!+\!1\}$.
By Lemma~\ref{proof 2.5}, we have either $Y(h,r)\leq F\!+\!1$ or $Y(h,r)\geq F\!+\!1$. Therefore, $Y(h,r)=[F\!+\!1,\infty)\cap\mathbb{Z}$.
Since $F \in Y(h)$ and $F\notin Y^*_h$, there exists an action $a \in \overline{A^*_h}$ s.t.\ $F \in Y(h,a)$ and $|Y(h,a)|>1$.
Since $Y(h,a)\subseteq\{F,F\!+\!1\}$, we have that $Y(h,a)=\{F,F\!+\!1\}$.
Therefore, since $F\!+\!1\in Y^*_h$, by gradual revelation and strategyproofness every $R_i\in\mathcal{R}_i(h)$ that chooses $a$ at $h$ has peak $\le F$.
In particular, since $\mathcal{R}_i(h,a)\ne\emptyset$, there exists $R_i\in\mathcal{R}_i(h)$ that has peak $\le F$. Therefore, since $F \notin Y^*_h$, we have by Lemma~\ref{lemma: owned choices}
that no other action $a'\in A(h)\setminus\{a\}$ has $F\in Y(h,a')$. Therefore (and since for any $a'\in A(h)\setminus\{r\}$ we have $Y(h,a')\subseteq\{F,F\!+\!1\}$), we conclude that any action $a'\in A(h)\setminus\{a,r\}$ has $Y(h,a')=\{F\!+\!1\}$, and so $\overline{A^*_h}=\{a,r\}$, as required.
\end{proof}

\begin{proof}[Proof of Lemma~\ref{proof 5}]
Completely analogous the the proof of Lemma~\ref{proof 4}.
\end{proof}

\section{Proof of Theorem~\ref{theorem: quasilinear}}\label{app: quasilinear}

\begin{proof}[Proof of Theorem~\ref{theorem: quasilinear}]
Assume the contrary.
Call the goods $a$ and $b$ and the bidders $1$ and $2$.
We denote by $(v(a), v(b))$ the preferences induced by a valuation of $v(a)$ for good $a$ and $v(b)$ for good $b$.

\medskip

We first consider the possible outcomes when the preferences of each bidder belong to the set $\{(20,2),(8,8),(6,6)\}$. (For ease of
presentation, we make no attempt to choose particularly low valuations or to tighten the analysis of utilities below.)

If both bidders have the same preferences $(v(a),v(b))$, then we claim that the utility of each bidder is at most $2$. (Under VCG with the Clarke
pivot rule, the utility of each bidder would be $0$.) Since in this case any allocation is welfare maximizing, we reason by cases.
\begin{itemize}
\item
If bidder~$1$ is awarded both goods and bidder~$2$ is awarded no good, then bidder~$2$ is charged nothing. Assume for contradiction that
bidder~$1$ is charged less than $v(a)+v(b)-2$. Then, bidder $1$ with preferences $(v(a)-1,v(b)-1)$ (with the preferences of bidder~$2$
unchanged), who would get no good if declaring his true preferences (and thus have utility~$0$), would rather misrepresent himself as
$(v(a),v(b))$ and get both goods (and have strictly positive utility), in contradiction to strategyproofness. So, the utility of bidder~$1$ is
at most $v(a)+v(b)-(v(a)+v(b)-2)=2$. The case in which bidder~$2$ is awarded both goods and bidder~$1$ is awarded no good is analyzed
analogously.
\item
If bidder~$1$ is awarded $a$ and bidder~$2$ is awarded $b$, then assume for contradiction that bidder~$1$ is charged less than $v(a)-1$. In
this case, bidder $1$ with preferences $(v(a)-1,0)$ (with the preferences of bidder~$2$ unchanged), who would get no good if declaring his true
preferences (and thus have utility $0$), would rather misrepresent himself as $(v(a),v(b))$ and get good $a$ (and have strictly positive
utility), in contradiction to strategyproofness. So, the utility of bidder~$1$ is at most $v(a)-(v(a)-1)=1$. Similarly, bidder~$2$ is charged
at least $v(b)-1$ and so has utility at most $1$.
The case in which bidder~$2$ is awarded $a$ and bidder~$1$ is awarded $b$ is is analyzed analogously.
\end{itemize}

If bidder~$1$ has preferences $(20,2)$ and bidder~$2$ has preferences $(8,8)$, then bidder~$1$ gets good $a$ and bidder~$2$ gets good $b$. In
this case, bidder~$1$ is charged at most $9$. Indeed, to see this note that if bidder~$1$ had preferences $(9,2)$ (with the preferences of
bidder~$2$ unchanged), then he would also get good $a$, but by strategyproofness would be charged at most $9$ as otherwise his utility would be
negative and he would have rather misrepresented his preferences to be $(0,0)$ for a utility of $0$. Therefore, since bidder~$1$ with preferences
$(20,2)$ has no incentive to misrepresent his preferences as $(9,2)$, he is charged at most $9$ as well.
Similarly, bidder~$2$ is charged at most $3$. (Under VCG with the Clarke pivot rule, they would be charged $8$ and $2$, respectively.) Similarly,
if bidder~$1$ has preferences $(8,8)$ and bidder~$2$ has preferences $(20,2)$, then bidder~$1$ is charged at most $3$ and bidder~$2$ is charged
at most~$9$.

If bidder~$1$ has preferences $(20,2)$ and bidder~$2$ has preferences $(6,6)$, then bidder~$1$ gets good $a$ and bidder~$2$ gets good $b$. In
this case, similarly to the above analysis, bidder~$1$ is charged at most $7$ and bidder~$2$ is charged at most $3$. (Under VCG with the Clarke
pivot rule, they would be charged $6$ and $2$, respectively.) Similarly, if bidder~$1$ has preferences $(6,6)$ and bidder~$2$ has preferences
$(20,2)$, then bidder~$1$ is charged at most $3$ and bidder~$2$ is charged at most $7$.

Finally, If bidder~$1$ has preferences $(8,8)$ and bidder~$2$ has preferences $(6,6)$, then bidder~$1$ gets both goods and bidder~$2$ gets no
good (and is charged nothing). In this case, bidder~$1$ is charged at most $14$.
Indeed, to see this note that if bidder~$1$ had preferences $(7,7)$ (with the preferences of bidder~$2$ unchanged), then he would also get both
goods, but by strategyproofness would be charged at most $14$ as otherwise his utility would be negative and he would have rather misrepresented
his preferences to be $(0,0)$ for a utility of $0$. Therefore, since bidder~$1$ with preferences $(8,8)$ has no incentive to misrepresent his
preferences as $(7,7)$, he is charged at most $14$ as well. (Under VCG with the Clarke pivot rule, bidder~$1$ would be charged $12$.) Similarly,
if bidder~$1$ has preferences $(6,6)$ and bidder~$2$ has preferences $(8,8)$, then bidder~$1$ is charged nothing and bidder~$2$ is charged at
most $14$.

\medskip

To see that the given social choice function cannot be implemented by a strategyproof mechanism, suppose it could, and let $\hat{M}$ be an
obviously incentive compatible gradual revelation mechanism that implements it. We now prune $\hat{M}$ by restricting the preferences of both
bidders to only be in the set $\{(20,2),(8,8),(6,6)\}$, to obtain an obviously incentive compatible gradual revelation mechanism $M$ for the
restricted preference domain.

Assume without loss of generality that $P(\emptyset)=1$. Since the choice of bidder~$1$ at $\emptyset$ is real (i.e., more than one action
exists), then since bidder~$1$ has only three possible preferences, some action $a$ is chosen by a single preference of bidder~$1$, and so
reveals his preferences (any other preference for bidder~$1$ chooses an action different from $a$). We reach a contradiction by reason by cases.

If action $a$ is chosen by the preferences $(20,2)$, then the minimal utility bidder~$1$ can get by reporting truthfully is at most $2$ (when
bidder~$2$ has preferences $(20,2)$ as well). Nonetheless, the maximal utility bidder~$1$ can get by deviating by misreporting his preferences as
$(8,8)$ is at least $22-14=8>2$ (when bidder~$2$ has preferences $(6,6)$). This contradicts obvious strategyproofness.

If action $a$ is chosen by the preferences $(8,8)$, then the minimal utility bidder~$1$ can get by reporting truthfully is at most $2$ (when
bidder~$2$ has preferences $(8,8)$ as well). Nonetheless, the maximal utility bidder~$1$ can get by deviating by misreporting his preferences as
$(6,6)$ is at least $8-3=5>2$ (when bidder~$2$ has preferences $(20,2)$). This contradicts obvious strategyproofness.

Finally, if action $a$ is chosen by the preferences $(6,6)$, then the minimal utility bidder~$1$ can get by reporting truthfully is at most $2$
(when bidder~$2$ has preferences $(6,6)$ as well). Nonetheless, the maximal utility bidder~$1$ can get by deviating by misreporting his
preferences as $(8,8)$ is at least $6-3=3>2$ (when bidder~$2$ has preferences $(20,2)$). This contradicts obvious strategyproofness.
\end{proof}

\section{Proof of Theorem~\ref{theorem: houses}}\label{app: house-matching}

Before beginning our analysis of housing problems, we introduce some more notation, which helps simplify it. While $Y^*_h$ has been defined as the set of outcomes that agent $P(h)$ can enforce at $h$, we here somewhat abuse notation by considering $Y^*_h$ as the set of houses that agent $P(h)$ can choose to be matched with at $h$. So $o\in Y^*_h$ means that $P(h)$ as an action $a\in A(h)$ such that $P(h)$ is matched with~$o$ for any outcome $y$ associated with a terminal history $h'$ with $(h,a)\subseteq h'$. While $M(B)$ generally denotes the outcome of the mechanism $M$ given the behavior profile~$B$, we here let $M(B)(i)$ be the house that $i$ is matched with under the outcome $M(B)$. Since the agents' preferences over matchings are functions of their preferences over the houses they are matched with, and since only a few top- and bottom-ranked houses matter for many of our arguments, we write $R_i:o$ for an (arbitrary) preference $R_i$ that ranks house $o$ at the top. Similarly $R_i:o, o'$
 denotes a preference that ranks $o$ above $o'$, which it ranks in turn above all other houses. This definition is extended to lists of any
 length. Analogously, $R_i:\ldots,o$ denotes a preference that ranks $o$ at the bottom, and $R_i:o',\ldots,o$ denotes a preference that ranks $o'$ at the top and $o$ at the bottom.

\begin{lemma}\label{lemma: SBL is OSP}
Any social choice function implemented via sequential barter with lurkers is Pareto-optimal. Moreover, such a social choice function can
be OSP-implemented.
\end{lemma}

\begin{proof}[Proof of Lemma~\ref{lemma: SBL is OSP}]
As outlined in Section~\ref{sec: house-matching}.
\end{proof}

 Theorem \ref{theorem: revelation} shows that any OSP implementable rule is implementable via an obviously incentive compatible gradual revelation mechanism. So, to show the converse direction, it is sufficient to show that any Pareto optimal rule that is implementable via an obviously incentive compatible gradual revelation mechanism that corresponds to  sequential barter with lurkers. For the reminder of the proof,
let $M$ be an obviously incentive compatible gradual revelation mechanism that implements a Pareto optimal rule. Since $\mathcal{R}$ is finite, all histories in $M$ are finite.

\begin{definition}
Let $h$ be a history of $M$ and $i\in N$ an agent.
\begin{enumerate}
\item
$O_i(h)\eqdef\{o: M(\mathbf{T}(R))(i)=o\text{ for some }R\in \mathcal{R}(h) \}$ is the set of houses assigned to $i$ in any terminal history following on $h$.
\item
Agent $i$ is matched to house $o$ at $h$ if $ O_i(h)=\{o\}$.
\item
 $D_i(h)\eqdef\{o: o\in Y^*_{h'}\text{ for some }h'\subseteq h\text{ with }P(h')=i \}$ is the set of houses that agent $i$ can choose to be matched with at some subhistory  of $h$.
\item
$D^<_i(h)\eqdef\{o: o\in Y^*_{h'}\text{ for some }h'\subsetneq h\text{ with }P(h')=i \}=D_i(h')$ for $h=(h',a)$ is the set of houses that agent $i$ can choose to be matched with at some \emph{strict} subhistory of $h$.
 \item
 If $P(h)=i$, then we say that agent $i$ moves at $h$. If $P(h')=i$ for a strict subhistory $h'$ of $h$, then we say that
 $i$ moves before $h$.
 \end{enumerate}
\end{definition}

\begin{remark}
$O_i(h)$ is (weakly) decreasing in $h$, while $D_i(h)$ and $D^<_i(h)$ are both (weakly) increasing in $h$. $D^<_i(h)$ is always a subset of $D_i(h)$, and that $i$ moves at $h$ if $D_i(h)$ and $D^<_i(h)$ differ.
\end{remark}

\begin{lemma}\label{lemma: one-undetermined}
Fix  a history $h$ of $M$ and let $i=P(h)$.
\begin{enumerate}
\item
There exists at most one action $a\in \overline{A^*_h}$ --- denote it by $\tilde{a}$.
\item
The behavior induced by the truthtelling strategy for a preference $R_i\in\mathcal{R}_i(h)$ at $h$ is as follows. Let $o=\max_{R_i} O_i(h)$.
\begin{itemize}
\item
If $o\in Y^*_h$, then there is an action $a\in A^*$ such that $\{R_i\}=\mathcal{R}_i(h,a)$, and furthermore, $O_i(h,a)=\{o\}$. Choose this action.
\item
Otherwise, the action $\tilde{a}$ exists. Choose this action.
\end{itemize}
\end{enumerate}
\end{lemma}

Note that Lemma \ref{lemma: one-undetermined} implies that $\mathcal{R}_i(h,\tilde{a})$ equals $\{R_i\in \mathcal{R}_i(h)\mid \max_{R_i}O_i(h)\notin Y^*_h\}$. Moreover, an action other than $\tilde{a}$ is chosen by some agent if and only if by choosing this action, this agent is matched with their most preferred house among all houses they could possibly be matched with at $h$; each such agent fully reveals her preference to the mechanism when making this choice.

\begin{proof}[Proof of Lemma~\ref{lemma: one-undetermined}]
We start with a preliminary observation, which states a sufficient condition for $\overline{A^*_h}$ to be a singleton.

\bigskip

(*) Fix $h,i$ such that $P(h)=i$. If for all $o,o'\in O_i(h)\setminus D_i(h)$ with $o\neq o'$ there exists some $R_i\in \mathcal{R}_i(h)$ with $R_i: o, o'$, then $\overline{A^*_h}$ is a singleton.

\bigskip

Suppose that Observation (*) does not hold, that is, suppose there exist two different actions $a$ and $a'$ in $\overline{A^*_h}$ under the condition of the observation. Since~$M$ is a gradual revelation mechanism,
$\mathcal{R}_i(h,a)\neq \emptyset$. Say $R_i\in \mathcal{R}_i(h,a)$ and let $o\eqdef\max_{R_i}O_i(h)$.
Since $o\in O_i(h)$, we have by strategyproofness that $o\in O_i(h,a)$.
Since $a\in \overline{A^*_h}$, there exists
 a profile $R\in\mathcal{R}(h,a)$ such that $M(\mathbf{T}(R))(i)\neq o$. Therefore, $o\notin D_i(h)$, for otherwise $\mathbf{T}_i(R_i)$  would not be a best reply to  $\mathbf{T}_{-i}(R_{-i})$. By the same reasons, there exists a
$R'_i\in \mathcal{R}_i(h,a')$ with $o'=\max_{R'_i}O_i(h)\in O_i(h,a')\setminus D_i(h)$. Since $o'\notin D_i(h)$, we have by gradual revelation that $o'\notin Y^*_h$, and therefore, by Lemma~\ref{lemma: owned choices} we have that $o'\notin O_i(h,a)$. In particular, $o\neq o'$. By assumption, $\mathcal{R}_i(h)$ therefore contains  a preference $R''_i: o, o'$. The truthtelling strategy for $R''_i$
cannot prescribe $a$ at $h$, as for any behavior that chooses $a$ at $h$, agent $i$ may be matched (i.e., is matched, for some $R_{-i}$) with a house in $O_i(h,a)\setminus \{o\}$, while for some behavior that chooses $a'$ at $h$, agent $i$ may be matched (i.e., is matched, for some, possibly different, $R'_{-i}$) with $o'$, which he strictly $R''_i$-prefers to each house in $O_i(h,a)\setminus\{o\}$. So the truthtelling strategy for $R''_i$ must prescribe an action $a''\neq a$ at $h$. But since $o \notin D_i(h)$, it follows that for any behavior that chooses $a''$ at $h$, agent $i$ may be matched with an outcome other than~$o$, while for some behavior that chooses $a$ at $h$, agent $i$ may be matched with $o$. There is, in sum, no obviously incentive compatible strategy for $R''_i$.

\bigskip

We now fix an agent $i$ and prove both parts of the lemma by full induction over histories $h$ with $P(h)=i$. If $h^*$ is minimal among the set of such histories, then $\mathcal{R}_i(h^*)=\mathcal{R}_i$ and the
 condition in
Observation (*) (and therefore, Part 1) is satisfied.
 We next show that Part 2 holds at a history $h$ if Part 1 does. Finally, we argue that the condition in Observation (*) holds at $h$ if Parts 1
 and 2 hold for all strict subhistories $h'\subsetneq h$ with $P(h')=i$.

\bigskip

Assume that Part 1 holds at a history $h$ with $P(h)=i$. Let $R_i\in\mathcal{R}_i(h)$ and let $o=\max_{R_i}O_i(h)$. If $o\in Y^*_h$, then by definition there exists an action $a'$ with $O_i(h,a')=\{o\}$. The truthtelling strategy for $R_i$ must prescribe an action $a$ with $O_i(h,a)=\{o\}$: by gradual revelation, any behavior that does not immediately (i.e., by the action it specifies at $h$) enforce that $i$ is matched with $h$, does not enforce it as a behavior as well; that is, when playing any such behavior, agent $i$ may be matched (i.e., is matched, for some $R_{-i}$) to an $R_i$-inferior house than when choosing $a'$ at $h$. Therefore, the truthtelling strategy for $R_i$ must prescribe at $h$ an action $a$ with $O_i(h,a)=\{o\}$.
Since $M$ is a gradual revelation mechanism, we have $\mathcal{R}_i(h,a)=\{R_i\}$.

Now suppose that $o\notin Y^*_h$. By Part 1, it is enough to show that the truthtelling strategy for $R_i$ prescribes an action $a\in\overline{A^*_h}$ (and so, by Part 1, $a=\tilde{a}$).
Suppose that the truthtelling strategy for $R_i$ prescribes an action $a\in A^*_h$, so $O_i(h,a)=\{o'\}$ for some $o'\in Y^*_h$, so $o'\ne o$. Since $o\in O_i(h)$,
by gradual revelation there must exist some $R'\in\mathcal{R}(h)$ with  $M(\mathbf{T}(R'))(i)=o$, so the $R_i$-best possible house when choosing $a'=\mathbf{T}(R'_i)(h)$ at $h$ (namely $o$) is strictly $R_i$-preferred to  the $R_i$-worst possible house when choosing $a$ at $h$ (namely $o'\ne o$), and the choice of $a$ at $h$ is not obviously strategyproof for $R_i$.

\bigskip

Finally, assume that Parts 1 and 2 hold for all strict subhistories $h'\subsetneq h$ with $P(h')=i$.
The condition in Observation (*) is trivially satisfied if   $O_i(h)\setminus D_i(h)$ is a singleton, so suppose that $|O_i(h)\setminus D_i(h)|\ge2$.
Therefore, by the induction hypothesis for Part 1, agent $i$ chooses $\tilde{a}$ in all strict subhistories of $h$.
Let $o,o'\in O_i(h)\setminus D_i(h)$ be distinct, and let
$R_i:o,o'$; since $D_i(\cdot)$ and $O_i(\cdot)$ are respectively weakly increasing and weakly decreasing, we have that $O_i(\cdot)\setminus D_i(\cdot)$ is weakly decreasing. Since $o\in O_i(h)\setminus D_i(h)$, we therefore have that $o\in O_i(h')\setminus D_i(h')$ for all strict subhistories $h'\subsetneq h$. Therefore, by the induction hypothesis for Part 2, the truthtelling strategy for $R_i$ prescribes the choice of $\tilde{a}$ at every such $h'$. Therefore, $R_i\in\mathcal{R}_i(h)$, and the sufficient condition cited in Observation (*) holds.
\end{proof}

\begin{definition}\label{def-lurker-force-dictator}
Let $h$ be a history of $M$ and let $i\in N$.
\begin{enumerate}
\item
Agent $i$ is \emph{active} at $h$ if $D_i(h)\ne\emptyset$ and $|O_i(h)|>1$.
\item
Agent $i$ can force a house $o$ at $h$ if he has a behavior for which he is matched with $o$ if $h$ is reached. (I.e.,
 there exists $R_i\in \mathcal{R}_i(h)$ such that $o=M(\mathbf{T}(R))(i)$ holds for all $R_{-i}\in\mathcal{R}_{-i}(h)$.)
\item
Agent $i$  is a \emph{dictator} at $h$ if he can force every house in $O_i(h)$ at $h$. (I.e., if $R\in\mathcal{R}(h)$, then $M(\mathbf{T}(R))(i)$ depends only on~$R_i$.)
\end{enumerate}
\end{definition}

\begin{remark}
No part of Definition~\ref{def-lurker-force-dictator} assumes that $i=P(h)$.
\end{remark}

\begin{remark}
If $i$ is a dictator at $h$, then it holds that
$O_i(h)=\{\max_{R_i}\!O_i(h)\mid \text{ for some }R_i\in \mathcal{R}_i(h)\}$.
\end{remark}

\begin{lemma}\label{lemma: can-force-d}
Fix a history $h$ of $M$ and an agent $i$ who has moved before or at $h$.
\begin{enumerate}
\item
If $o=\max_{R_i}O_i(h)\in D_i(h)$ holds for some $R_i\in\mathcal{R}_i(h)$,  then $i$ can force $o$  at $h$.
\item If $i$ is not a dictator at $h$, then $D_i(h)\subsetneq O_i(h)$.
\end{enumerate}
\end{lemma}

\begin{proof}[Proof of Lemma~\ref{lemma: can-force-d}]\leavevmode
\begin{enumerate}
\item Since $o\in D_i(h)$, there is a history $h'\subseteq h$ with $P(h')=i$ and $o\in Y^*_{h'}$. Since $i$ may deviate  to a behavior that chooses an action $a$ at $h'$ with $O_i(h',a)=\{o\}$, agent $i$ with the given preference $R_i$ must weakly prefer $M(\mathbf{T}(R))(i)\in O_i(h)$ to $o$ for any $R_{-i}\in\mathcal{R}_{-i}(h)$. Therefore, since $o=\max_{R_i}O_i(h)$,
we have that $M(\mathbf{T}(R))(i)=o$ for any $R_{-i}\in\mathcal{R}_{-i}(h)$, and so agent $i$ can force $o$ at $h$.
\item
Since $i$ is not a dictator at $h$, there exists $R_i\in\mathcal{R}_i(h)$ such that $i$ cannot force at $h$ the house
$o=\max_{R_i}O_i(h)$. By Part 1, therefore $o\notin D_i(h)$. In particular, $O_i(h)\ne D_i(h)$. It therefore remains to show that $D_i(h)\subseteq O_i(h)$. Let $d\in D_i(h)$.
Consider the preference $R'_i:o,d$. Since $o\in O_i(h)\setminus D_i(h)$, by Lemma~\ref{lemma: one-undetermined} we have that $R'_i\in \mathcal{R}_i(h)$.
Since $i$  cannot force the house $o$ at $h$, there exist preferences $R_{-i}\in\mathcal{R}_{-i}(h)$ such that $M(\mathbf{T}(R'_i,R_{-i}))(i)\neq o$.
Since $M$ is strategyproof and since $d\in D_i(h)$, agent $i$ must weakly $R'_i$-prefer $M(\mathbf{T}(R'_i,R_{-i}))(i)$ to $d$. Since $R'_i: o,d$, we have that $M(\mathbf{T}(R'_i,R_{-i}))(i)=d$, and therefore $d\in O_i(h)$.\qedhere
\end{enumerate}
\end{proof}

\begin{definition}
Let $h$ be a history of $M$.
\begin{itemize}
\item
An agent $i$ who has moved before $h$, but is not a dictator at $h$ is a \emph{lurker} for a house $o$ at $h$ if the maximal
\emph{strict} subhistory $h'$ of $h$ with $P(h')=i$ is such that $o\in O_i(h')$ and $D_i(h')=O_i(h')\setminus\{o\}$.
\item
$O(h)$ is the union of the sets  $O_i(h)$ for all agents $i$ that are not yet matched at $h$.
\item
$G(h)$ is the set of houses in $O(h)$ that have no lurkers at $h$.
\end{itemize}
\end{definition}

\begin{lemma}\label{lemma: lurker-iff}
Fix a history $h$ in $M$, an agent $i$ who is not yet matched at $h$, and a house $o\in O_i(h)$ that $i$ does not lurk at $h$. Then there
exists $R_i\in\mathcal{R}_i(h)$ with $R_i:\ldots,o$.
\end{lemma}

\begin{proof}[Proof of Lemma~\ref{lemma: lurker-iff}]
If $i$ does not move at any strict subhistory of $h$, then $\mathcal{R}_i(h)=\mathcal{R}_i$ and the claim trivially holds.

If $i$ is a dictator at $h$, then $O_i(h)=\{\max_{R_i}O_i(h)\mid R_i\in \mathcal{R}_i(h)\}$.
Since $i$ is not yet matched, then $|O_i(h)|>1$. Therefore, there exists $R'_i\in\mathcal{R}_i(h)$ whose top choice among $O_i(h)$ is some $o'\ne o$. Let $R_i$ be the preference obtained from $R'_i$ by demoting $o$ to be least preferred. Since $i$ is not yet matched at $h$ and since $o'\in O_i(h)$ is ranked by $R_i$ strictly above $o$ (since $o\in O_i(h)$), by Lemma~\ref{lemma: one-undetermined}, we have that $R_i\in \mathcal{R}_i(h)$ as well.

Finally, if $i$ moves at some strict subhistory of $h$ and is not a dictator, then let $h'\subsetneq h$ be the maximal strict subhistory of $h$ with $P(h')=i$.
Since $i$ is not a dictator at $h$, she is not a dictator at $h'$, and so by Lemma~\ref{lemma: can-force-d}(2), $D_i(h')\subsetneq O_i(h')$.
Since $i$ is neither a dictator nor a lurker for $o$ at $h$, there exists $o'\in O_i(h')\setminus (D_i(h')\cup\{o\})$. By
Lemma~\ref{lemma: one-undetermined},  $R_i\in\mathcal{R}_i(h',\tilde{a})=\mathcal{R}_i(h)$ for every $R_i:o'$, and in particular for $R_i:o',\ldots,o$. (Note that
$(h',\tilde{a})$ is indeed a subhistory of $h$ since $i$ is not yet matched at $h$.)
\end{proof}

\begin{lemma}\label{lemma: not-in-o-lurked}
Fix a history $h$ of $M$. If an agent $i$ is not yet matched at $h$, then  $G(h)\subseteq{O_i(h) \cup D^<_i(h)}$. If an agent $i$ is a lurker at $h$ then $G(h)\subseteq D^<_i(h)$.
\end{lemma}

\begin{proof}[Proof of Lemma~\ref{lemma: not-in-o-lurked}]
Fix a house $o\in G(h)$ such that $o\notin D^<_i(h)$. It is enough to show that $o\in O_i(h)$. Let $R'_i\in\mathcal{R}_i(h)$, and let $R_i$ be the preference obtained from $R'_i$ by promoting $o$ to be most preferred. Since $i$ is not yet matched at $h$ and since $o\notin D^<_i(h)$, by Lemma~\ref{lemma: one-undetermined} we have that $R_i\in\mathcal{R}_i(h)$ as well. By Lemma~\ref{lemma: lurker-iff}, for every not-yet-matched agent $j\ne i$ with $o\in O_j(h)$, there exists $R_j\in\mathcal{R}_j(h)$ with $R_j:\ldots,o$. Fix some $R_k\in\mathcal{R}_k(h)$ for every other agent~$k$ (i.e., matched agent or agent with $o\notin O_k(h)$; note that matched agents are not matched to $o$, since $o\in G(h)\subseteq O(h)$). Note that $R\in\mathcal{R}(h)$. Since $R_i$ ranks $o$ at the top while any agent $j\ne i$ with $o\in O_j(h)$ ranks $o$ at the bottom, the Pareto optimality of~$M$ implies that $M(\mathbf{T}(R))(i)=o$, and therefore $o\in O_i(h)$.

If $i$ lurks some house $o_i$ at $h$ then by definition, $O_i(h)=\{o_i\}\cup D^<_i(h)$. Since by the first claim $G(h)\subseteq O_i(h)\cup D^<_i(h)=\{o_i\}\cup D^<_i(h)$, but by definition $o_i\notin G(h)$, the second claim follows.
\end{proof}

\begin{lemma}\label{lemma: force-lurked-then-force-g}
Let $h$ be a history of $M$. Let $i$ be an active agent at $h$ such that some $o_l\in D_i(h)$ has a lurker at $h$. Then for every $o\in G(h)$, either $o\in D_i(h)$ or $i$ can force
$o$  at $h$.
\end{lemma}

\begin{proof}[Proof of Lemma~\ref{lemma: force-lurked-then-force-g}]
It suffices to show that  $i$ can force any $o \in G(h)\setminus D_i(h)$ at $h$.
Since $o \in G(h)\setminus D_i(h)$, by Lemma~\ref{lemma: not-in-o-lurked} we have that $o\in O_i(h)$. Since $o \notin D_i(h)$ and $o_l\in
D_i(h)$, we have that $o\ne o_l$.
Let $l$ be the lurker of $o_l$; since $o_l\in D_i(h)$, we have that $l\ne i$.
Lemma~\ref{lemma: not-in-o-lurked} and $o\in G(h)$ imply $o \in D_l(h)$.
Since $o\in O_i(h)\setminus D_i(h)$, we have that $R_i:o,o_l\in\mathcal{R}_i(h)$.
We will show that $M(\mathbf{T}(R))(i)=o$ holds for every $R_{-i}\in \mathcal{R}_{-i}(h)$.
Since $l$ is a lurker for $o_l$ at $h$, we have that $o_l=\max_{R_l}O_l(h)$. By Lemma~\ref{lemma: one-undetermined}, also $R'_l:o_l,o\in\mathcal{R}_i(h)$.
Since $o_l\in D_i(h)$ and $o\in D_l(h)$, $i$ and $l$ respectively must $R_i$-prefer and $R'_l$-prefer their matches
$M(\mathbf{T}(R'_l,R_{-l}))(i)$ and $M(\mathbf{T}(R'_l,R_{-l}))(l)$
to $o_l$ and $o$. Since $M$ is Pareto optimal, $M(\mathbf{T}(R'_l,R_{-l}))(i)= o$ and  $M(\mathbf{T}(R'_l,R_{-l}))(l)=o_l$. Since $M$ is strategyproof, $M(\mathbf{T}(R))(l)=o_l$ must also hold. Since $M(\mathbf{T}(R))(i)R_i o_l$ and $M(\mathbf{T}(R))(i)\neq  o_l$, we have that $M(\mathbf{T}(R))(i)= o$ and $i$ can force $o$ by playing $R_i$.
\end{proof}

\begin{lemma}\label{lemma: force-g-then-dictator}
Let $h$ be a history of $M$, and let $i$ be an active nonlurker at $h$.
If for every $o\in G(h)$, either $o\in D_i(h)$ or $i$ can force $o$ at $h$, then $i$ is a dictator at $h$.
\end{lemma}

\begin{proof}[Proof of Lemma~\ref{lemma: force-g-then-dictator}]
Assume that $i$ is not a dictator at $h$. Therefore, there exists $o_l\in O_i(h)\setminus D_i(h)$ that $i$ cannot force at $h$.
By assumption, $o_l\notin G(h)$, implying that $o_l$ has a lurker, say $l$, at $h$. Since $i$ is not a lurker, $l\neq i$.

If $P(h)\ne i$ let $h'=h$; otherwise, let $h'=(h,\tilde{a})$. (Since $i$ is not a dictator at $h$, then $\tilde{a}\in\overline{A^*_h}$ and $(h,\tilde{a})$ is a nonterminal history of $M$.) By gradual revelation, $k\eqdef P(h')\neq i$.
By Lemma~\ref{lemma: one-undetermined}, $Y^*_{h'}\ne\emptyset$, so there exists $o^* \in Y^*_{h'}$. We conclude the proof by considering three cases, in each case reaching a contradiction by showing that $k$ cannot force $o^*$ at $h'$.

Case 1: If $o^*\in G(h)$, then by assumption, either $o^*\in D_i(h)$ or $i$ can force $o^*$ at $h$. If $o^*\in D_i(h)$, Lemma~\ref{lemma:
can-force-d}(2) implies that $o^*\in O_i(h)$. Since $l$ is a lurker at $h$ and $o^*\in G(h)$, by Lemma~\ref{lemma: force-lurked-then-force-g}, we
have that $o^*\in D_l(h)$.
Fix $R_{-k}\in\mathcal{R}_{-k}(h')$ with $R_i: o_l,o^*$ and $R_l: o_l, o^*$.
By strategyproofness, and since either $i$ can force $o^*$ at $h'$ or $o^*\in D_i(h')$, we obtain that $M(\mathbf{T}(R))(\{i,l\})=\{o,o^*\}$ for every $k\in\mathcal{R}_k(h)$, so $k$ cannot force $o^*$ at $h'$ --- a contradiction.

Case 2: If $o^*=o_l$, then let $o'\in G(h)$.
Fix $R_{-k}\in\mathcal{R}_{-k}(h')$ with $R_l: o_l, o'$ and $R_i: o_l,o'$.
By the argument of Case 1 with $o^*$ replaced by $o'$, we have that $M(\mathbf{T}(R))(\{i,l\})=\{o_l,o'\}$ for every $R_k\in \mathcal{R}_k(h')$, so $k$ cannot force $o_l=o^*$ at $h'$  --- a contradiction.

Case 3: Finally, if $o^*\notin G(h)$ and $o^*\ne o_l$, then $o^*$ has a lurker, say $l^*$, at $h$. Let $o'\in G(h)$. Note that $o^*,o_l,o'$ are all distinct.
Fix $R_{-k}\in\mathcal{R}_{-k}(h')$ with $R_i: o_l,o'$, $R_l: o_l,o'$, and $R_{l^*}: o^*,o'$.
By the argument of Case 2, we have that $M(\mathbf{T}(R))(\{i,l\})=\{o_l,o'\}$ for every $R_k\in\mathcal{R}_k(h')$.
Since $l^*$ is a lurker at $h$ and $o'\in G(h)$, by Lemma~\ref{lemma: not-in-o-lurked} we have that $o'\in D_{l^*}(h)$. So $M(\mathbf{T}(R))(l^*)$ must by strategyproofness be $R_{l^*}$-preferred to $o'$. Since $o'$ is matched to either $l$ or $i$, then $M(\mathbf{T}(R))(l^*)$ must equal $o^*$, so $k$ cannot force $o^*$ at $h'$  --- a contradiction.
\end{proof}

The next Lemma combines the two preceding ones to show that any agent $i$ who has a lurked house  $o_l$ in his set $D_i(h)$  is a dictator at $h$.

\begin{lemma}\label{lemma: force-lurked-then-dictator}
Let $h$ be a history of $M$. Let $i$ be a nonlurker at $h$ such that some $o_l\in D_i(h)$ has a lurker at $h$. Then $i$ is a dictator at $h$.
\end{lemma}

\begin{proof}[Proof of Lemma~\ref{lemma: force-lurked-then-dictator}]
Assume that $i$ is not a dictator at $h$, so since $o_l\in D_i(h)$, agent $i$ is active at $h$.
For every $o\in G(h)$, by Lemma \ref{lemma: force-lurked-then-force-g} either $o\in D_i(h)$ or $i$ can force $o$  at $h$. 
Since $i$ is an active nonlurker at $h$, $i$ is by Lemma \ref{lemma: force-g-then-dictator} a dictator at $h$.
\end{proof}

\begin{definition}
For the remainder of the proof, we assume without loss of generality that the set of lurkers at a history $h$ ($h$ will be clear from context) is $L\eqdef\{1,\ldots,\lambda\}$, where $l<l'$ holds for $l,l'\in L$ if and only if $l$ becomes a lurker at a subhistory of the history at which $l'$ becomes a lurker (where both are subhistories of $h$). Moreover, we denote the house lurked by lurker $l$ by~$o_l$.
\end{definition}

\begin{lemma}\label{lemma: lurkers-hierarchy}
Let $h$ be a history of $M$.
\begin{enumerate}
\item
$o_1,\ldots,o_\lambda$ are distinct.
\item
$O_{l}(h)=O(h)\setminus\{o_1,\ldots,o_{l-1}\}$ holds for all $l\in L$.
\end{enumerate}
\end{lemma}

\begin{proof}[Proof of Lemma~\ref{lemma: lurkers-hierarchy}]
We prove the claim by induction over $l\in L$.
Assume that the claim holds for the lurkers $1,\ldots,l\!-\!1$ (and their associate houses) and let $h_l\subsetneq h$ be the history at which $l$ becomes a lurker. (So $P(h_l)=l$, and agent $l$ is not a lurker at $h_l$, but is a lurker at $(h_l,\tilde{a})$.)
Since $l$ is a lurker at $(h_l,\tilde{a})$, Lemma \ref{lemma: not-in-o-lurked} implies that $G(h_l)\setminus\{o_l\}=G(h_l,\tilde{a})\subseteq D^<_l(h_l,\tilde{a})=D_l(h_l)$.

Assume for contradiction that $o_l\in\{o_1,\ldots,o_{l-1}\}$. Therefore, $G(h_l)=G(h_l)\setminus\{o_l\}\subseteq D_l(h_l)$, and so $l$, who is an active nonlurker at $h_l$, is by Lemma~\ref{lemma:
force-g-then-dictator} a dictator at $h_l$ (and its superhistory $h$) --- a contradiction. In sum, $o_l$ is distinct from $o_1,\ldots, o_{l-1}$.

Recall that $O(h)\setminus\{o_1,\ldots,o_l\}=G(h_l)\setminus\{o_l\}\subseteq D_l(h_l)$. 
If this containment is strict, then $o_k\in D_l(h_l)$ for some lurker $k<l$, and therefore, by Lemma \ref{lemma: force-lurked-then-dictator}, agent $l$ is a dictator at $h_l$ --- a contradiction.
So, $D_l(h_l)=O(h)\setminus\{o_1,\ldots,o_l\}$. 
Since $D_l(\cdot)$ and $O(\cdot)$ are respectively weakly increasing and weakly decreasing, we have $D_{l}(h)= O(h)\setminus\{o_1,\ldots,o_{l}\}$ as well,
and since $l$ is still a lurker at $h$, we have that $O_l(h)=D_l(h)\cup\{o_l\}=O(h)\setminus\{o_1,\ldots, o_{l-1}\}$, as required.
\end{proof}

\begin{lemma}\label{lemma: how-to-unravel}
Let $h$ be a history of $M$ with $i=P(h)\notin L$.
Let $R\in\mathcal{R}(h)$ with $\mathbf{T}(R)(i)=a\in A^*_h$. Let $o$ be the unique house such that $O_i(h,a)=\{o\}$.
\begin{itemize}
\item
 $M(\mathbf{T}(R))(l)$ can for each lurker $l\in L$ be calculated by the following inductive procedure, starting with Step $1$ and $o'\gets o$.

Step $l$:
\begin{itemize}
\item
If $o'\in G(h)$, then:
\begin{enumerate}
\item
$M(\mathbf{T}(R))(k)=o_k$ for each $l\le k\le\lambda$.
\end{enumerate}
\item
Else:
\begin{enumerate}
\item
Let $l'\in L$ be the unique lurker such that $o_{l'}=o'$.
\item
$M(\mathbf{T}(R))(k)=o_k$ for each $l\leq k< l'$
\item
$M(\mathbf{T}(R))(l')=\max_{R_{l'}}D_{l'}(h)$.
\item
If $l'<\lambda$, then go to Step $l'+1$ with $o'\gets M(\mathbf{T}(R))(l')$.
\end{enumerate}
\end{itemize}
\item
$M(\mathbf{T}(R))(\{i,1,\ldots, \lambda\})=\{o_1,\ldots, o_{\lambda}, o'\}$, where $o'\in G(h)$. If $o\in G(h)$, then $o'=o$.
\end{itemize}
\end{lemma}

\begin{proof}[Proof of Lemma~\ref{lemma: how-to-unravel}]
To prove Part 1, assume that in the steps preceding step $l$, lurkers $\{1,\ldots,l\!-\!1\}$ (this set is empty if $l=1$) are matched in accordance with the above procedure.
We will show that the matches described in step $l$ are correct. Recall that $o'$ (as it is defined for step $l$) has already been matched to some agent in $\{i,1,\ldots,l\!-\!1\}$.

If $o'\in G(h)$, then fix some $l\le k\le\lambda$. We must show that $k$ is matched with $o_k$. By Lemma~\ref{lemma: not-in-o-lurked}, $G(h)\subseteq D_{k}(h)$, and so $o'\in D_k(h)$.
If it is the case that $R_k=R'_k:o_k,o'$, then by strategyproofness $k$ must be matched with a house she weakly $R_k$-prefers to $o'$, meaning that $k$ must be matched at $(h,a)$ with $o_k$. Since $k$ lurks $o_k$ at $h$, we have by Lemma~\ref{lemma: one-undetermined} that $R'_k\in\mathcal{R}_k(h)=\mathcal{R}_k(h,a)$.
By strategyproofness and since $k$ lurks $o_k$ at $h$, we therefore have that $k$ must be matched with $o_k$ at $(h,a)$ regardless of whether or not $R_k=R'_k$.

It remains to consider the case in which $o'\notin G(h)$, i.e., $o'=o_{l'}$ for for some $l'\in L$. By the preceding steps, $l'\ge l$. (Trivial if $l=1$, and otherwise follows from the fact that $o_{l'}=o'\in D_{l-1}(h)$ and from Lemma~\ref{lemma: lurkers-hierarchy}.)
For every lurker $l\le k<l'$, we have that $o'=o_{l'}\in D_k(h)$ and by an argument similar to the one in the preceding paragraph, we have that $k$ is matched with $o_k$ at $(h,a)$.
Since $o_{l'} \notin O_{l'}(h,a)$, we have by strategyproofness that lurker $l'$ must be matched at $(h,a)$ with his most preferred house $o^*$ in $D_{l'}(h)$, as required.

Since there are finitely many lurkers and since at least one lurker is matched at each step, the process terminates after finitely many steps. The second Part is a direct consequence of the first.
 \end{proof}

\begin{lemma}\label{lemma: force-lurked-terminator}
Let $h$ be a history of $M$ with lurked houses. Let $t=P(h')$ for $h'$ the maximal superhistory of $h$ of the form $h'=(h,\tilde{a},\tilde{a},\ldots,\tilde{a})$.
Then the agent $t$ is not a lurker at $h$. Furthermore,
If $o_l\in O_k(h)$ for some
nonlurker $k$ and a lurked house $o_l$ at $h$, then $k=t$.
Moreover, $D_t(h')=O(h)$.
\end{lemma}

\begin{proof}[Proof of Lemma~\ref{lemma: force-lurked-terminator}]
We start by proving the second statement.
Say that
 $o_l\in O_k(h)\setminus G(h)$ holds for some nonlurker $k$.
Let $l$ be the lurker of $o_l$ at $h$; by assumption, $l\ne k$.
Since $o_l\in O_k(h)$, there exists some $R\in \mathcal{R}(h)$ with $M(\mathbf{T}(R))(k)=o_l$.

If $R\notin\mathcal{R}(h')$, then at some $h''\subsetneq h'$, the agent $i=P(h'')$ chooses an action $\mathbf{T}_i(R_i)(h'')=a\ne\tilde{a}$. Since $a\ne\tilde{a}$, the choice of $a$ by $i$ matches $i$ with some house $o'$. Since $i$ is not a dictator at $h'$, she can by Lemma \ref{lemma: force-lurked-then-dictator} only choose a house in $G(h)$ at $h$, and so $o'\in G(h)$.
By Lemma \ref{lemma: how-to-unravel}, following the choice of $a$ by $i$ at $h$, house $o_l$ is matched with agent $l$ rather than with agent $k$ --- a contradiction, so $R\in\mathcal{R}(h')$. 

Let $a=\mathbf{T}_t(R_t)(h')$ be agent $t$'s choice at $h'$.
Note that $a\ne\tilde{a}$ by definition of $h'$. Therefore, by Lemma \ref{lemma: one-undetermined}, agent $t$ becomes matched to some house $o$ at $(h',a)$.
By Lemma~\ref{lemma: how-to-unravel}, $\{o_1,\ldots, o_\lambda\}\subset M(\mathbf{T}(R))(t,1,\ldots,\lambda)$. Since $1,\ldots,\lambda$ are lurkers while $k$ is a not a lurker, we have that $k=t$.

The first statement is a direct consequence of the second one: if $t$ is a lurker at $h$, then no nonlurker $k$ at $h$ may be possibly matched with any house that is lurked at $h$, so all lurked houses, under any superhistory of $h$, will be matched with lurkers, so by Pareto optimality each lurker gets her lurked house, so all lurkers are dictators at $h$ --- a contradiction.

Finally, we show that $D_t(h')=O(h)$. Recall that the oldest lurker at $h$, agent 1, lurks house~$o_1$. Since $1$ is a lurker, there is a profile of preferences $R\in \mathcal{R}(h)$ with $M(\mathbf{T}({R}))(1)\neq o_1$. By Lemma~\ref{lemma: lurkers-hierarchy}, $o_1$ is not contained in $ O_{l}(h)$ for any lurker $l>1$. Since, by the first statement, $t$ is the only nonlurker with $o_1\in O_t(h)$,  $A(h')$ must contain an action $a$
with  $O_t(h,a)=\{o_1\}$, so $o_1\in D_t(h')$.

If $o_1\in D^<_t(h')$, then by the second statement and by Lemma~\ref{lemma: force-lurked-then-dictator}, $t$ was given the option to choose $o_1$ at a strict subhistory of $h'$ before $1$ became a lurker (since he was not a dictator when given this option). By the same lemma, $t$ became a dictator immediately when $1$ became a lurker, and so at that point $t$ was able to force all remaining houses in $O(h)$ that she could not force before. Since the next move of $t$ after that is a dictatorial one, it is at $h'$, and so $D_t(h')=O(h)$, as required. Assume henceforth, therefore, that $o_1\notin D^<_t(h')$, and assume for contradiction that there exists a house $o'\in O(h)\setminus D_t(h')$.

Since $1$ is the oldest lurker at $h$ and since $o'\ne o_1$, Lemma \ref{lemma: lurkers-hierarchy}
implies $o'\in D_1(h)$.
For every $R\in \mathcal{R}(h')$ with
$R_t:o',o_1$ (this preference indeed reaches $h'$ since $o'\notin D_t(h')$ and $o_1\in O_t(h')\setminus D^<_t(h')$) and $R_1:o_1,o'$, by strategyproofness $t$~is matched with $o_1$ (since $o'\notin D_t(h')$ and $t$ is a dictator at $h'$, however $o_1\in D_t(h')$) while $1$ is matched with $o'$ (since $o_1$ is matched to $t$, however $o'\in D_1(h)$)--- a contradiction to
Pareto optimality.
\end{proof}

\begin{lemma}\label{lemma: no-three}
Let $h$ be a history of $M$ s.t.\ two distinct players $i,j$ are active nonlurkers at $h$, and $k=P(h)\notin\{i,j\}$ has $D^<_k(h)=\emptyset$.
Then, either $i$ or $j$ is a dictator at $h$.
\end{lemma}

\begin{proof}[Proof of Lemma~\ref{lemma: no-three}]
Assume that all agents at $h$ are either active or have not yet moved at $h$. (If not, then
remove from $M$ all players that have moved before $h$ but are no longer active at $h$, along with their
houses, keeping all else equal, to obtain a new incentive compatible gradual revelation mechanism for the smaller set of agents and houses.)

Assume for contradiction that neither $i$ nor $j$ is a dictator at $h$.
Therefore, by Lemmas~\ref{lemma: can-force-d}(2), \ref{lemma: force-g-then-dictator} and \ref{lemma: force-lurked-then-dictator}, $D_i(h) \subsetneq G(h) \subseteq O_i(h)$ and $D_j(h)\subsetneq G(h) \subseteq O_j(h)$.
Furthermore, by Lemma~\ref{lemma: force-lurked-terminator} at most one of the two agents can be matched with any lurked houses, so we have either $G(h)=O_i(h)$ have or $G(h)=O_j(h)$ (or both).
By Lemma~\ref{lemma: one-undetermined}, and since the choice of $k$ at $h$ is nontrivial, there exists $o \in Y^*_h=D_k(h)$. We reason by cases.

\medskip

Case 1: $D_i(h) \cup D_j(h) = G(h)$. Assume w.l.o.g.\ that $O_i(h)=G(h)$. We will obtain a contradiction by showing that $j$ can force any $o_1 \in G(h)\setminus D_j(h)$ at
$h$ (and therefore, by Lemma~\ref{lemma: force-g-then-dictator}, is a dictator).
Let $R_j:o_1$. Since $o_1\in G(h)\setminus D_j(h) \subseteq O_j(h)\setminus D_j(h)$, we have by Lemma~\ref{lemma: one-undetermined} that $R_j\in\mathcal{R}_j(h)$.
Fix $R_{-j}\in\mathcal{R}_{-j}(h)$; it is enough to show that $M(\mathbf{T}(R))(j)=o_1$.

Since $o_1 \in G(h)\setminus D_j(h)$, we have that $o_1 \in D_i(h)$. By Lemma~\ref{lemma: one-undetermined}, $R_i$ ranks $o_2$ highest
among $O_i(h)=G(h)$ for some $o_2\in G(h)\setminus D_i(h)$. Therefore, $o_2 \in D_j(h)$. Let $R'_i: o_2, o_1$ and $R'_j:o_1, o_2$. Note that since $o_2\in O_i(h)\setminus D_i(h)$, and since $o_1\in O_j(h)\setminus D_j(h)$, we have that $R'_i\in \mathcal{R}_i(h)$ and $R'_j\in \mathcal{R}_j(h)$.
By strategyproofness, $i$ and $j$ must respectively $R'_i$- and $R'_j$-prefer $M(\mathbf{T}(R'_{\{i,j\}},R_{-\{i,j\}}))(i)$ and $M(\mathbf{T}(R'_{\{i,j\}},R_{-\{i,j\}}))(j)$ to $o_1$ and $o_2$. By Pareto optimality, therefore $M(\mathbf{T}(R'_{\{i,j\}},R_{-\{i,j\}}))(i)=o_2$ and $M(\mathbf{T}(R'_{\{i,j\}},R_{-\{i,j\}}))(j)=o_1$.
By strategyproofness, $M(\mathbf{T}(R'_j,R_{-j}))(i)=o_2$ as well. Since $j$ must $R'_j$-prefer $M(\mathbf{T}(R'_j,R_{-j}))(j)$ to $o_2$, we have that 
$M(\mathbf{T}(R'_j,R_{-j}))(j)=o_1$. By strategyproofness, therefore $M(\mathbf{T}(R))(j)=o_1$, as required. So, $j$ can force $o_1$ at $h$ --- a contradiction.

\medskip

Assume henceforth, therefore, that $D_i(h) \cup D_j(h) \subsetneq G(h)$.
By Lemma~\ref{lemma: not-in-o-lurked}, we have that $O_k(h) \supseteq G(h)$.
Assume for now that $o\in G(h)$.

\medskip

Case 2: There exist $o_1 \in D_i(h) \setminus D_j(h)$, $o_2 \in D_j(h) \setminus D_i(h)$, and $o\in Y^*_h\subseteq G(h)\setminus (D_i(h) \cup
D_j(h))$. Since $M$ is obviously strategyproof, by the pruning argument of \citet{Li2015}, it is obviously strategyproof on any restricted domain.
We restrict the domain of preferences as follows: Each agent $i,j,k$ ranks $\{o_1,o_2,o\}$ above all other houses. The three agents $i,j,k$ have a common (known) ranking of all other houses. Each other agent
$\ell$ has a fixed preference $R_\ell\in\mathcal{R}_\ell(h)$ that ranks $o_1,o_2,o$ at the bottom (note that  each $\mathcal{R}_\ell(h)$ contains such a preference since no such agent $\ell$ lurks  $o_1$, $o_2$, or $o$).

Obtain a new mechanism by pruning all choices that are never taken. Since the  preferences of all agents other than $i$, $j$, and $k$ are known in the restricted domain, none of these agents has any choice to make. After condensing all nodes that proved no choice, we obtain a new (not necessarily gradual revelation) obviously strategyproof and Pareto optimal mechanism $\hat{M}$ in which only $i$, $j$, and $k$ move.  The history $\hat{h}$ in the new tree that corresponds to $h$ in the original tree is such that $D_i(\hat{h})=\{o_1\}$, $D_j(\hat{h})=\{o_2\}$, and  $D_k(\hat{h})=\{o\}$, and $O_i(\hat{h})=O_j(\hat{h})=O_k(\hat{h})=\{o_1,o_2,o\}$.

Fix $R\in \mathcal{R}(\hat{h})$ with $R_i: o_2, o_1$,  $R_j: o_1, o_2$. Since $D_i(\hat{h})=\{o_1\}$, $D_j(\hat{h})=\{o_2\}$, and since $\hat{M}$ is strategyproof and Pareto optimal, $\hat{M}(\mathbf{T}(R))(i)=o_1$ and $\hat{M}(\mathbf{T}(R))(j)=o_2$. This implies in particular that $k$ cannot force $o_1$ at $\hat{h}$. Fix $R'\in\mathcal{R}(\hat{h})$ with $R'_i: o, o_2, o_1$ and $R'_j: o, o_2, o_1 $ and $R'_k: o_1$. We have $M(\mathbf{T}(R'))(k)=o_1$ since $\hat{M}$ is Pareto optimal. Since $k$ cannot force $o_1\in O_k(\hat{h})$, but may prefer $o_1$ most, we have that $k$ is not a dictator at $\hat{h}$ and has by
Lemma~\ref{lemma: one-undetermined} a unique  action $\tilde{a}\in A(\hat{h})$ that does not determine his match (i.e., the action $\tilde{a}\in A(h)$ was not removed during the pruning process).

Assume that $i=P(\hat{h},\tilde{a})$ (the analysis for the cases $P(\hat{h},\tilde{a}) \in \{j,k\}$ is analogous). By Lemma \ref{lemma: one-undetermined}, $\mathcal{R}_i(\hat{h},\tilde{a})=\{R_i\mid \max_{R_i}\{o_1,o_2,o\}\neq o_1\}$. Again by Lemma \ref{lemma: one-undetermined}, since $D^<_i(\hat{h},\tilde{a})=D_i(\hat{h},\tilde{a})=o_1$ we have that agent $i$ must have an action $a$ at $(\hat{h},\tilde{a})$ that forces some $o'\in\{o_2, o\}$. Let $R_i\in\mathcal{R}_i(\hat{h},\tilde{a},a)$. For $R_j: o, o_2, o_1\in\mathcal{R}_j(\hat{h},\tilde{a},a)$ and $R_k: o_2, o, o_1\in\mathcal{R}_k(\hat{h},\tilde{a},a)$,
we nonetheless have by strategyproofness that $M(\mathbf{T}(R))(\{k,j\})=\{o_2,o\}$ , so $i$ playing $a$ does not force $o'$ --- a contradiction.

\medskip

Case 3: There exists $o^*\in D_i(h)\cap D_j(h)$, and $o \in G(h)\setminus (D_i(h)\cup D_j(h))$. Consider a profile $R$ with $R_i: o, o^*$
and $R_j: o, o^*$. We note that any such profile reaches $h$. By strategyproofness, $M(\mathbf{T}(R))(\{i,j\})=\{o,o^*\}$, so $k$ cannot force $o$ at $h$
 --- a contradiction.

\medskip

Assume henceforth, therefore, that $o \in D_i(h) \cup D_j(h)$.

\medskip

Case 4: $o \in D_i(h) \cap D_j(h)$. Recall that there exists $o' \in G(h) \setminus (D_i(h) \cup D_j(h))$. Consider a profile $R$
with $R_i:o',o$ and $R_j:o',o$. We note that any such profile reaches $h$. By
strategyproofness, $M(\mathbf{T}(R))(\{i,j\})=\{o',o\}$, so $k$ cannot force $o$ at $h$ --- a contradiction.

\medskip

Case 5: $o \in D_i(h) \setminus D_j(h)$ and there exists $o_2 \in D_j(h) \setminus D_i(h)$. (The case $o \in D_j(h) \setminus D_i(h)$ and there
exists $o_1 \in D_i(h) \setminus D_j(h)$ is analyzed analogously.)

Consider a profile $R$ with $R_i:o_2, o$  and $R_j: o, o_2$.  We note any such profile reaches $h$.
 By strategyproofness, $M(\mathbf{T}(R))(\{i,j\})=\{o_2,o\}$, so $k$ cannot force $o$ at $h$ --- a contradiction.

\medskip

Case 6: $o \in D_i(h) \setminus D_j(h)$ and, $D_i(h)\supseteq D_j(h) \cup Y^*_h$, and there exists $o_2 \in D_j(h) \setminus Y^*_h$. (The case
$o
\in D_j(h) \setminus D_i(h)$, $D_j(h) \supseteq D_i(h)\cup Y^*_h$ and there exists $o_1 \in D_i(h) \setminus Y^*_h$ is analyzed analogously.) We
will show that $i$ can force any $o' \in G(h) \setminus D_i(h)$ via some behavior from $h$, therefore, by Lemma~\ref{lemma: force-g-then-dictator}, obtaining a contradiction to the
assumption
that $i$ is not a dictator at $h$.
Assume therefore, for contradiction, that there exists a behavior profile $R_{-i}\in\mathcal{R}_{-i}(h)$ s.t.\ for no $R_i\in \mathcal{R}_i(h)$ does
$i$ get $o'$. Against this $R_{-i}$, let $i$ play the following behavior: always pass (i.e.,~$\tilde{a}$) unless $i$ can force $o'$ via an action $a$ with $Y(a)=\{o'\}$, in which case play $a$; if $i$ becomes a dictator without the option to force $o'$, then force an arbitrary house. By gradual revelation, this behavior corresponds to some preference $R_i$. By assumption, at some history along $\Path(\mathbf{T}(R))$, it would no longer hold that $o' \in O_i$.
Let $h'$ be the maximal history along this path with $o'\in O_i(h')$.

We first claim that at no history $h\subseteq h''\subseteq h'$ is it possible that any agent other than $i$ has a possible action that forces a house $o''\notin D_i(h'')$.
Indeed, otherwise let $h''$ be such a minimal history and w.l.o.g.\ $P(h'')\ne k$. By Lemma~\ref{lemma: force-lurked-then-force-g},
we can assume w.l.o.g.\ that $o''\in G(h'')$. We obtain a contradiction as in Case 5: since the behavior so far of each of $i$ and $k$ is consistent with the preference that ranks
$o''$ first (since $o'' \in G(h'')\setminus D_i(h'')$ and by minimality of $h''$ also $o'' \in G(h'')\setminus D_k(h'')$ since $D_k(h'')\subseteq D_i(h'')$) and ranks $o \in D_i(h)\cap D_k(h) \subseteq D_i(h'') \cap D_k(h'')$ second, then by strategyproofness they will together get $o''$ and $o$ if these really are their preferences, and so $j$ cannot force $o''$ at $h''$.

Furthermore, at no history $h\subseteq h''\subseteq h'$ does any agent actively force (according to $\mathbf{T}(R)$) a house $o'' \in D_i(h'')$, as otherwise by strategyproofness $i$ would get $o'$ since his preference up until $h''$ is consistent with the preference $R'_i:o',o''$.

Combining the claims from the two preceding paragraphs, we get that every action played by any agent in any history $h\subseteq h''\subseteq h'$ is pass (i.e., $\tilde{a}$).
In particular, $(h',\tilde{a})\in\Path(\mathbf{T}(R))$.
Since by definition of $h'$, we have $o'\notin O_i(h',\tilde{a})$, we have that $i\ne t$, for $t$ as defined in Lemma~\ref{lemma: force-lurked-terminator} w.r.t.\ the history $h'$.

We now claim that no agent, other than possibly $i$, lurks $o'$ at $h'$. Assume that some agent $l\ne i$ lurks $o'$ at $h'$.
Since up until $h'$, no house is offered to any agent before it is offered to $i$, and since $i$ is not a dictator at $h'$ (since $o'$ is possible for $i$ at $h'$ but $i$ cannot force it), we have that $i$ is also a lurker for $o'$ at $h'$ (since as $i\ne t$, by Lemma~\ref{lemma: force-lurked-terminator} if $i$ is a nonlurker at $h'$ then $O_i(h')=G(h')\subseteq O_l(h')$) --- a contradiction to the distinctness of lurked houses.

Recall that $o'\in O_i(h')$ but $o'\notin O_i(h',\tilde{a})$.
Since we have shown that no agent can force $o'$ at $h'$ and no agent other than possibly $i$ lurks $o'$ at $h'$, we have that $o'\in O(h',\tilde{a})$. By Lemma~\ref{lemma: not-in-o-lurked} and since $o'\notin O_i(h',\tilde{a})\cup D_i(h',\tilde{a})$, we have that $o'\notin G(h',\tilde{a})$, and so some agent $l\ne i$ is a lurker for $o'$ at $(h',\tilde{a})$. Once again, this means that $i$ is already (since $i\ne P(h')$) a lurker for $o'$ at $h'$ --- a contradiction. (By Lemma~\ref{lemma: force-lurked-terminator}, it is possible that the agent $t\ne l$ gets $o'$, so we have a contradiction to strategyproofness since $i$ and $l$ are then promised their respective ``second best'' options, which may coincide when $t$ gets $o'$.)

\medskip

Finally, we consider the case in which $o\notin G(h)$. Recall that $D^<_k(h)=\emptyset$. By definition, there is a lurker $\ell\ne k$ for $o$ at
$h$, and so by Lemma~\ref{lemma: force-lurked-then-force-g}, $k$ can force any house $o_1\in G(h)=O_i(h)$ from $h$, and so by the analysis of
either Case 4 or Case 5 we are done.
\end{proof}

Lemma ~\ref{lemma: no-three} shows that at most two active agents at any $h$ are lurkers. The next Lemma describes the relationship between the sets of houses that these two agents can be endowed with at any given history $h$.

\begin{lemma}\label{lemma: other-lurker-cant-force-from-terminator}
Let $h$ be a history of $M$ with lurked houses  and let $t$ be as in Lemma~\ref{lemma: force-lurked-terminator}. If $k=P(h)\ne t$ is a nonlurker
and $t$ is not a dictator at $h$, then $Y^*_h \cap D_t(h)=\emptyset$.
\end{lemma}

\begin{proof}[Proof of Lemma~\ref{lemma: other-lurker-cant-force-from-terminator}]
Let $o'\in D_t(h)$. We will show that $o'\notin Y^*_h$.
Recall that the oldest lurker at $h$, agent 1, lurks house~$o_1$. 
Since $t$ is not a dictator at $h$, we have by Lemma~\ref{lemma: force-lurked-then-dictator} that $o_1\notin D_t(h)$, and so
$o_1\ne o'$. Since $1$ is the oldest lurker at $h$, we have that
$o' \in O(h)\setminus \{o_1\}=D_1(h)$. By Lemmas~\ref{lemma: force-lurked-terminator} and~\ref{lemma: can-force-d}(2), $o_1\in O_t(h)$. Let $R_1:
o_1, o'$ and $R_t: o_1, o'$. Since
$o'\in D_t(h)\cap D_1(h)$ while $o_1\in (O_1(h)\setminus D_1(h))\cap(O_t(h)\setminus D_t(h))$, we have by strategyproofness that when playing
$R_1$
and $R_t$ (which reach $h$) from $h$, $1$ and $t$ get $o_1$ and $o'$, and so $k$ cannot force $o'$ at $h$ and so $o' \notin
Y^*_h$, as
required.
\end{proof}

The characterization is a direct consequence of the above lemmas: since if an agent $i$ is a dictator at some history $h$, then we can have him
``divulge'' his entire preferences at that point, and thus he becomes inactive, and therefore there always exists a mechanism implementing the
same social choice function where no ``already-active'' dictators are still active while others make nondictatorial moves (or become active).

As long as two active agents $i,j$ are nonlurkers, then no one else can join the game and $i$ and $j$ slowly accumulate houses into $D_i$ and
$D_j$, respectively. If one of them, say $i$, becomes a lurker, then her lurked house becomes impossible for all agents except one, which we
denote by $t$ (for terminator), and from that point no agent can accumulate houses that are in $D_t$. Since $i$ became a lurker, another player
$k$ can join the game, and if one of $k$ and $j$ becomes a lurker, then another can join the game, etc. At some point, the terminator $t$ will be
a dictator, and after her choice, all lurkers by strategyproofness will become matched via a dictatorship from the ``oldest'' lurker $i$ (who has
the largest list of possible houses) to the ``youngest'' lurker (who has the smallest list of possible houses), and no lurkers remain. If no one
forced a house from the option list $D_{\ell}$ of the nonlurker that did not claim a house, then he continues (along with $D_{\ell}$) to the next
round.

\subsection{Matching with Outside Options}\label{homeless-bosses}

When agents may prefer being unmatched over being matched to certain houses (and possibly also more agents exist than houses), the OSP-implementable and Pareto optimal social choice functions are similar in spirit, but considerably more detailed to describe, due to the following phenomenon: Recall that at any point when an agent becomes a lurker, the sets of possible houses $O_i$ of some other agents $i$ are reduced. In this case, such agents $i$ may be asked by the mechanism whether they prefer being unmatched over being matched with any house in (the reduced) $O_i$ (note that for this reason, the set $O_i$ has to be tracked and updated for all agents and not only for active agent $i\in T$). If such an agent indeed prefers being unmatched, then she can divulge her full preferences, and based on these preferences, the mechanism may now reduce $O_i$ for some other agent, etc.\ (thus, in a sense, dynamically choosing the ``terminator'' agent $t$ defined in the proof above), or even award some lurkers their lurked houses. The analysis for this case is similar, yet more cluttered, and we omit it.

\end{document}